\documentclass[a4paper]{article}
\usepackage[english]{babel}
\pdfoutput=1
\usepackage{jheppub_mod}

\usepackage{amsmath}
\usepackage{amssymb}
\usepackage{graphicx,color}

\usepackage{amsmath,amsthm}
\usepackage{txfonts}

\usepackage{mathrsfs}
\usepackage{ulem}

\usepackage{bm}

\PassOptionsToPackage{square,numbers}{natbib}
 \usepackage{natbib}

\newcommand{\R}{\mathbb{R}}

\newcommand{\so}{\mathfrak{so}}

\newcommand{\D}{\mathrm{d}}
\newcommand{\E}{\mathrm{e}}
\newcommand{\I}{{\rm i}}

\newcommand{\Tr}{\text{Tr}}
\newcommand{\ctg}{\mathrm{ctg}}

\newcommand{\Su}{\mathrm{SU}(2)}
\newcommand{\su}{\fs\fu(2)}

\newcommand{\SO}{\mathrm{SO}(3)}
\newcommand{\SU}{\mathrm{SU}(2)}
\newcommand{\bsigma}{{\boldsymbol\sigma}}
\newcommand{\btau}{{\boldsymbol\tau}}
\newcommand{\btheta}{{\boldsymbol\theta}}

\def\be{\begin{eqnarray}}
\def\ee{\end{eqnarray}}

\newcommand{\fg}{\mathfrak{g}}

\newcommand{\fs}{\mathfrak{s}}  
  
\newcommand{\fu}{\mathfrak{u}}

\newcommand{\Ad}{\mathrm{Ad}}
\newcommand{\ad}{\mathrm{ad}}

\newcommand{\sgn}{\mathrm{sgn}}

\newcommand{\rmE}{{\mathrm{E}}}
\newcommand{\rmS}{{\mathrm{S}}}
\newcommand{\rmH}{{\mathrm{H}}}
\newcommand{\rmT}{{\mathrm{T}}}

\newcommand{\bR}{{\mathbf{R}}}
\newcommand{\bO}{{\mathbf{O}}}
\newcommand{\bH}{{\mathbf{H}}}
\newcommand{\bJ}{{\mathbf{J}}}
\newcommand{\Gr}{{\mathrm{Gram}}}

\renewcommand{\tilde}{\widetilde}
\newcommand{\oast}{\circledast}

\newtheorem{theorem}{Theorem}
\newtheorem{lemma}{Lemma}
\usepackage{amsthm}

\usepackage{yhmath}

\title{\fontsize{.67cm}{1em}\selectfont Encoding Curved Tetrahedra in Face Holonomies:\\ \vspace{9pt}\fontsize{.5cm}{1em}\selectfont  a Phase Space of Shapes from Group-Valued Moment Maps}
\author[1]{Hal M. Haggard, }
\author[2]{Muxin Han, }
\author[3]{and Aldo Riello }

\affiliation[1]{Physics Program, Bard College, Annandale-on-Hudson, NY 12504, USA}
\affiliation[2]{Institut f\"ur Quantengravitation, Friedrich-Alexander Universit\"at Erlangen-N\"urnberg, Staudtstra{\ss}e 7/B2, 91058\\ Erlangen, Germany}
\affiliation[3]{Perimeter Institute for Theoretical Physics, 31 Caroline Street North, Waterloo, Ontario Canada N2L 2Y5}

\emailAdd{hhaggard(AT)bard(DOT)edu}
\emailAdd{muxin.han(AT)gravity(DOT)fau(DOT)de} %
\emailAdd{ariello(AT)perimeterinstitute(DOT)ca} 

\abstract{ We present a generalization of Minkowski's classic theorem on the reconstruction of tetrahedra from algebraic data to homogeneously curved spaces. Euclidean notions such as the normal vector to a face are replaced by Levi-Civita holonomies around each of the tetrahedron's faces. This allows the reconstruction of both spherical and hyperbolic tetrahedra within a unified framework. A new type of hyperbolic simplex is introduced in order for all the sectors encoded in the algebraic data to be covered. Generalizing the phase space of shapes associated to flat tetrahedra leads to group valued moment maps and quasi-Poisson spaces. These discrete geometries provide a natural arena for considering the quantization of gravity including a cosmological constant. A concrete realization of this is provided by the relation with the spin-network states of loop quantum gravity. This work therefore provides a bottom-up justification for the emergence of deformed gauge symmetries and quantum groups in 3+1 dimensional covariant loop quantum gravity in the presence of a cosmological constant. }


\arxivnumber{}

\begin{document}

\maketitle%\vspace{-7mm}

%-----------------------------------------------------------------------------------------
\section{Introduction}\label{sec_Intro}

In 1897, Hermann Minkowski proved a reconstruction theorem stating that to each non-planar polygon with $L$ edges  $\{\vec a_\ell \in \R^3, \ell\in\{1,\dots,L\} | \sum_\ell \vec a_\ell = \vec 0 \}$, one can associate a unique convex polyhedron in Euclidean three-space $\rmE^3$ with $L$ faces. The area and outward pointing normal of its $\ell$-th  face are $|\vec a_\ell|$ and $\vec a_\ell / |\vec a_\ell|$, respectively \cite{Minkowski1897,Alexandrov2005}. One hundred years later, in 1996, Michael Kapovich and John J. Millson showed how the space of polygons with fixed edge lengths admits a natural phase space structure \cite{Kapovich1996}. The combination of these results is remarkable: it points out that discrete geometries are a natural arena for dynamics. One may then wonder whether this arena is related to the theory of dynamic geometry {\it par excellence}, general relativity. The answer turns out to be positive, though not in a trivial way. In fact, the Kapovich-Millson phase space can be quantized via geometric quantization techniques \cite{Conrady2009} and this quantized space gives a compelling interpretation of the Hilbert space of loop quantum gravity \cite{Rovelli2007,Thiemann2004} (restricted to a single graph) in terms of discrete quantum geometries \cite{Bianchi2011PRD,Freidel2010twist}. The main notions behind this are the following. In loop quantum gravity, the fundamental phase-space variables are $\su$ fluxes (momenta) and $\Su$ holonomies (coordinates) carried by the Faraday-Wilson lines of the gravitational field. These lines cross at nodes, where $\Su$ gauge invariance is imposed as a momentum conservation equation, often referred to as the Gau{\ss} or closure constraint
\be
\sum_{\ell=1}^L (\vec \tau_\ell)^R = \vec 0\,,
\ee
here $\ell$ labels the $L$ Faraday-Wilson lines at a node (all supposed outgoing), and $(\vec \tau_\ell)^R$ is the right invariant vector field on the $\ell$-th copy of $\Su$, i.e. the flux operator along the $\ell$-th Faraday-Wilson line. Since the norm of the flux of the gravitational field carried by one of these lines is associated to the area it carries \cite{Rovelli1995NPB}, it is physically meaningful to reinterpret this equation using Minkowski's theorem. In this way, it can be read as the definition of a quantum convex polyhedron at each intersection of $L$ gravitational Faraday-Wilson lines. How these polyhedra are glued to one another and how they encode the extrinsic geometry of the three-space they span is more complicated and we refer to the cited literature for more details. Nevertheless, the crucial point here is that to each kinematical state of loop quantum gravity one can associate a discrete piecewise-flat quantum geometry thanks to Minkowski's theorem.

In this paper we move toward the generalization of this construction to the case where the model space for the discrete geometry is curved instead of flat. In other words, we generalize Minkowski's theorem to tetrahedra whose faces are flatly embedded in the three-sphere $\rmS^3$ and hyperbolic three-space $\rmH^3$, and conjecture that a similar construction may hold for general curved polyhedra. From a purely mathematical point of view the generalization of Minkowski's theorem is interesting in its own right, and requires new inputs in order to replace ``the notion of face direction by some notion not relying on parallelism in the Euclidean sense'' (\cite{Alexandrov2005}, p. 346), or, in other words, to deal with the parallel transport of the face normals to a single base point. Moreover, the question arises whether the space of curved tetrahedra also admits a natural phase space structure, and eventually how close it is to the Kapovich-Millson one. We will show that a natural phase space structure exists, and it coincides with the one studied by Thomas Treloar in \cite{Treloar2000}.

Surprisingly, this phase space structure is the same in both the spherical and hyperbolic case, which have a unified description in our framework, and it is exactly the generalization of the Kapovich-Millson phase space to geodesic polygons  embedded in $\rmS^3$. (This $\rmS^3$ is \textit{not} the manifold in which the curved tetrahedron is embedded and, again, underlies both the positively- and negatively-curved cases.)  Beside pure mathematics, this generalization is relevant to physics as well. Indeed, this construction is thought to bear strong relations to quantum gravity in the presence of a cosmological constant. On the one hand, this is apparent through the requirement that the simplicial decomposition of the bulk geometry be a solution of Einstein's field equations with the cosmological term within each building block \cite{Bahr2009}. On the other, curved tetrahedra made an appearance already in the semiclassical limit of the Turaev-Viro state sum \cite{Turaev1992,Mizoguchi1992,Taylor2004,Taylor2005}, which in turn, is known to be related to Edward Witten's Chern-Simons quantization of three-dimensional gravity with cosmological constant \cite{Witten1988}. The relations with quantum gravity in (Anti-)de Sitter space constitute our main motivation \cite{HHKR}.

Several works, with close connections to ours, have focused on three spacetime dimensions.  Notably, in  \cite{Dupuis2013,Bonzom2014a,Bonzom2014b,Dupuis2014,Charles2015} and  \cite{Noui2011can,Pranzetti2014tv}, the precise connection was investigated between the Chern-Simons quantization of three dimensional gravity and the spinfoam or loop-theoretic polymer quantizations, respectively. With this in mind, we should emphasize that the present work studies three dimensional discrete geometries {\it as boundaries of four dimensional spacetimes}; this is in contrast to the research cited above, which focused on the description of geometries in two-plus-one dimensions. This difference in dimensionality implies a mismatch in the geometrical quantities encoded in the Faraday-Wilson lines: in four and three spacetime dimensions these carry units of area and length respectively. Unsurprisingly, the geometrical reconstruction theorems are completely different in the two cases.

 A second interesting divergence of the two approaches is the fact that the sign of the geometric curvature must be decided {\it a priori} in the two-plus-one case (in particular, Girelli {\it el al.} restrict their analysis to the hyperbolic case), while it is determined at the level of each solution in our case. Again, this is because our formalism automatically allows for---in fact, requires---both positively and negatively curved geometries. It is intriguing to attribute this difference to the lack of a local curvature degree of freedom in three-dimensional gravity (since these are purely kinematical constructions, one should take this statement {\it cum granu salis}). In spite of these differences, there is an important feature the two constructions share: in both cases one is naturally lead to consider  phase-space structures and symmetries that are deformed with respect to the standard ones of quantum gravity. In particular the momentum space of the geometry is curved and the symmetries  are distorted to (quasi-)Poisson Lie symmetries, which are the classical analogues of quantum-group symmetries. We leave the discussion of the quantization of our phase space and its symmetries for a future publication.
%{\bf  \textcolor{red}{leave comments on quantization for the conclusions} }\\

Another piece of recent work in the loop gravity literature that interestingly shares some features with our construction is by Bianca Dittrich and Marc Geiller \cite{Dittrich2014a,Dittrich2014b,Dittrich2015}. While constructing a new representation (a new ``vacuum'') for loop quantum gravity, adapted to describe states of constant curvatures (and no metric), they are lead to deal with exponentiated fluxes as the meaningful operators. As a consequence, areas in their formulation are also associated to $\Su$---instead of $\su$---elements. Nonetheless, the parallel seems limited, since it appears that they are not forced to deform their phase space and symmetry structures.

Finally, Kapovich and Millson, and later Treloar, recognized that the phase space structure of polygons corresponds  to William Goldman's symplectic structure on the moduli space of flat connections on  an $L$ times punctured two-sphere. Our result provides this correspondence with a more direct and physical interpretation. In fact, the punctures on the two-sphere can be understood as arising from the gravitational Faraday-Wilson lines piercing an ideal two-sphere surrounding one of their intersection points. Exactly as in the flat case, these lines characterize the face areas and define a polyhedron. This picture can be used to extend the covariant loop quantum gravity framework in four spacetime dimensions, the spinfoam formalism, to the case with a non-vanishing cosmological constant. This is the goal of our companion paper \cite{HHKR}, which proposes a generalization of the spinfoam models constructed by John Barrett and Louis Crane \cite{Barrett1998,Baez1999}, and which eventually  developed into the Engle-Pereira-Rovelli-Livine/Freidel-Krasnov (EPRL/FK) models \cite{Engle2007,Engle2008,Freidel2008}.

In the companion paper \cite{HHKR}, we analyze $\mathrm{SL}(2,\mathbb{C})$ Chern-Simons theory with a specific Wilson graph insertion. The model can be viewed as a deformation of the EPRL/FK model aimed at introducing the cosmological constant in the covariant loop quantum gravity framework. The semiclassical analysis of the model's quantum amplitude is given by the four-dimensional Einstein-Regge gravity, augmented by a cosmological term, and discretized on homogeneously curved four simplices. Interestingly, the key equation studied in this paper, i.e. the generalization of the Gau{\ss} (or closure) constraint to curved geometry, arises in that context simply as one of the equations of motion.

The present paper is divided in two parts. In the first we introduce the generalization of Minkowski's theorem to curved tetrahedra. First we discuss the strategy underlying the theorem in the spherical case (\autoref{sec_strategy}). We then gradually extend our analysis to more general settings eventually including both signs of the curvature (\autoref{sec_hypgeom} to \ref{sec_twosheeted}). The first part concludes with the statement and proof of the theorem in its general form (\autoref{sec_theorem}). The second part is dedicated to the description of the phase space of shapes of discrete tetrahedra. We first introduce the subject and its relations with curved polygons and moduli spaces of flat connections (\autoref{sec_relations}). Then, we review the quasi-Poisson structure one can endow $\Su$ with (\autoref{sec_qpoisson}), which serves as a preliminary step to the actual description of the phase space of shapes (\autoref{sec_tetrahedron_phsp}). We conclude this part with a brief account of the equivalent quasi-Hamiltonian approach (\autoref{sec_qhamiltonian}). The paper closes with some physical considerations and an outlook towards future developments of the work (\autoref{sec_summary} and \ref{sec_outlook}).

\newpage
%-----------------------------------------------------------------------------------------
\part{Minkowski's theorem for curved thetrahedra}
\section{General strategy and the spherical case}\label{sec_strategy}

In the flat case, Minkowski's theorem associates to any solution of the so-called ``closure equation''
\be
\sum_{\ell=1}^4 \vec a_\ell \equiv \vec a_1 + \vec a_2  + \vec a_3 + \vec a_ 4 = \vec 0\,,\quad\text{with } \vec a_\ell \in \R^3\,,
\label{eq_flatclosure}
\ee
a tetrahedron in $\rmE^3$, whose faces have area $a_\ell := |\vec a_\ell|$ and outward normals $\hat n_\ell := \vec a_\ell/a_\ell$. We call the vectors $\vec a_\ell \equiv a_\ell \hat n_\ell$  area vectors. 

Our generalization begins with the most symmetrical curved spaces, the three-sphere $\rmS^3$ and hyperbolic three-space $\rmH^3$, which have positive and negative curvature, respectively. 
%{\bf Later on, add as a note the fact that this can be seen really as a reconstruction theorem in the following sense: I test space along certain holonomies, which is the best model of the enclosed space I can reconstruct out of this? \textcolor{red}{did not find where to fit this...}} 
In these ambient spaces, define a curved polyhedron as the convex region enclosed by a set of $L$ flatly embedded surfaces (the faces) intersecting only at their boundaries. A flatly embedded surface is a surface with vanishing extrinsic curvature and the intersection of two such surfaces is necessarily a geodesic arc of the ambient space. In the spherical case, the flatly embedded surfaces and the geodesic arcs are hence portions of great two-spheres and of great circles, respectively. Note that in the hyperbolic case this definition includes curved polyhedra extending to infinity. This and other properties specific to the hyperbolic case are discussed beginning in the next section.

%A useful way to visualize the spherical polyhedra, which will be most useful in the hyperbolic case, is from an extrinsic perspective. Take $\rmS^3$ to be the unit sphere in $\rmE^4$ carrying the metric induced by its flat ambient space. Then, the great sphere $\rmS^2_g(N)$ is defined by the intersection of a three-plane passing through the origin of $\rmE^4$ and normal to the unit four-vector $N\in \rmE^4$, with the unit three-sphere itself. Two such three-planes $N_1$ and $N_2$, such that $N_1\neq\pm N_2$, necessarily intersect along a two-plane passing through the origin and identified by the bivector $N_1\wedge N_2$. Such a plane intersects the unit three-sphere along the great circle $\rmS^1_g(N_1\wedge N_2)$. As an example, in this way one can partition the unit three-sphere into $2^4=16$ spherical tetrahedra, by considering four independent three-planes (i.e. such that $\det(N_1,N_2,N_3,N_4)\neq 0$). \autoref{fig_a} and \ref{fig_b} illustrate this construction in one dimension less.\\ 

%{\bf Fig of a plane through the origin of E3 intersecting a two sphere, and three great circles intersecting on a sphere and partitioning it into 2**3=8 triangles}\\

Our main result is that Minkowski's theorem and the closure equation, Eq. \eqref{eq_flatclosure}, admit a natural generalization to curved tetrahedra in $\rmS^3$ and $\rmH^3$. The curved closure equation is 
\be
O_4 O_3 O_2 O_1 = \E \,, \quad\text{with } O_\ell \in \SO,
\label{eq_closure}
\ee
where $\E$ denotes the identity in $\SO$. In the remainder of this section, we explain how this equation encodes the geometry of curved tetrahedra. Before going into this, we want to stress the essential non-commutativity of this equation, which mirrors the fact that the model spaces are curved, and therefore is the crucial feature of our approach. Indeed, non-commutativity has far-reaching consequences that are particularly apparent in the last section of the paper where the curved closure equation is used as a moment map. This non-commutativity will also be the source of an ambiguity in the reconstruction that is unique to the curved case.

As in the flat case, the variables appearing in the closure equation are associated to the faces of the tetrahedron. Indeed, the $\{O_\ell\}$ shall be interpreted as the holonomies of the Levi-Civita connection around each of the four faces of the tetrahedron. Since the faces of the tetrahedron are by definition flatly embedded surfaces in $\rmS^3$, any path contained within them parallel transports the local normal to the face at its starting point into the local normal to the face at its endpoint. Therefore, choosing at every point of the face a frame in which the local normal is parallel to $\hat z$, one can reduce via a pullback the $\so(3)$ connection to an $\mathfrak{so}(2)$ one  without losing any information. ({Note that there always exists a unique chart covering an open neighborhood of the whole face.}) In this two-dimensional setting, it is a standard result that a vector parallel transported around a closed (non self-intersecting) loop within the unit sphere gets rotated by an angle equal to the area enclosed by the loop. Therefore, the holonomy $O_\ell$ around the $\ell$-th face of the spherical tetrahedron, calculated at the base point $P$ contained in the face itself, is given by
\be
O_\ell(P) = \exp \left\{a_\ell \hat n_\ell(P) \cdot \vec J \;\right\}\,,
\label{eq_O}
\ee
where $\{\vec J\}$ are the three generators of $\so(3)$, $a_\ell$ is the area of the face, and $\hat n(P)\in \rmT_P\rmS^3$ is the direction normal to the face in the local frame at which the holonomy is calculated. Let us for a second ignore the issues related to curvature and non-commutativity, and discuss what happens in the flat Abelian limit where the radius of curvature of the three-sphere goes to infinity and the ambient space becomes nearly flat. To make this explicit, introduce the sphere radius $r$ into the previous expression: 
%{\bf change radius notation $R \rightarrow r$ also in the rest of the paper, notably last sections}
\be
O_\ell = \exp \left\{ \frac{a_\ell}{r^2} \hat n_\ell \cdot \vec J \;\right\}\,.
\ee 
In the limit $r\rightarrow\infty$, the curved closure equations reduces, at the leading order, to the flat one:
\be
O_\ell \stackrel{r\rightarrow\infty}{\approx} \E + \frac{a_\ell}{r^2}\hat n_\ell \cdot \vec J + \dots
\quad\Rightarrow\quad
O_4 O_3 O_2 O_1 \stackrel{r\rightarrow\infty}{\approx} \E + r^{-2} \Big(a_4\hat n_4 + a_3 \hat n_3 + a_2\hat n_2 + a_1 \hat n_1\Big) \cdot \vec J + \dots .
\ee
Importantly, the geometrical meaning of the variables is exactly the same
%\footnote{The direction of the normals - i.e. whether they are inward or outward pointing - will be discussed later in this section.} 
as in Eq. \eqref{eq_flatclosure}. Thus, our formulation subsumes the flat one as a limiting case.

%Curiously, the curved reconstruction theorem is conceptually simpler than its flat counterpart, in the sense that it does not make any ``active'' use of the information about the areas contained in the holonomies: the ambient space having an intrinsic curvature scale, it is possible to translate from the dihedral and face angles to the edge lengths of the tetrahedron. Nonetheless, despite this gained conceptual simplicity, the curved theorem is rather subtle, while its flat counterpart is completely trivial in the tetrahedral case.\\

In the curved setting, it is crucial to keep track of the holonomy base point; for within a curved geometry, only quantities defined at, or parallel transported to, a single point can be compared and composed with one another. Therefore, all four of the holonomies appearing in the curved closure equation must have the same base point. In spite of this, there is no point shared by all four faces of the tetrahedron at which one can naturally base the holonomies, and therefore at least one of them must be parallel transported away from its own face before being multiplied with the other three. Actually, the curved closure equation itself has no information about the base points of the holonomies or about which paths they have been parallel transported along to arrive at a common frame. This must be an extra piece of information that needs to be fed into the reconstruction algorithm. 
%In this respect, this choice is not different from attaching to trace classes of the holonomies the interpretation of areas, or from choosing $\rmS^3$ (or $\rmH^3$) as the ambient spaces for the reconstruction, instead of more complicated ones. 
Analogous interpretational choices---though for clear reasons less numerous---have to be made in the flat case. Here, we provide a standard set of paths on an abstract tetrahedron embedded in $\rmS^3$ along which the $\{O_\ell\}$ are assumed to be calculated. Such a choice of standard paths must also account for the presence of the identity element on the right hand side of the curved closure equation. This comes from the fact that the chosen standard paths compose to form a homotopically trivial loop. Interestingly, the curved closure equation can also be related to an integrated version of the Bianchi identities for the three dimensional Riemann tensor (see e.g. \cite{Freidel2003}).% Unfortunately, before continuing down this road, we need to introduce extra conventions to properly deal with orientations.\\

Label the vertices of the geometrical tetrahedron as in \autoref{fig_topordering}. This numbering induces a topological orientation on the tetrahedron, which must be consistent with the geometrical orientation of the paths around the faces. Faces are labeled via their opposite vertex (e.g., face 4 is the one at the bottom of \autoref{fig_topordering}), while edges are labeled by the two vertices they connect. %
\begin{figure}%
\begin{center}%
\includegraphics[height=3cm]{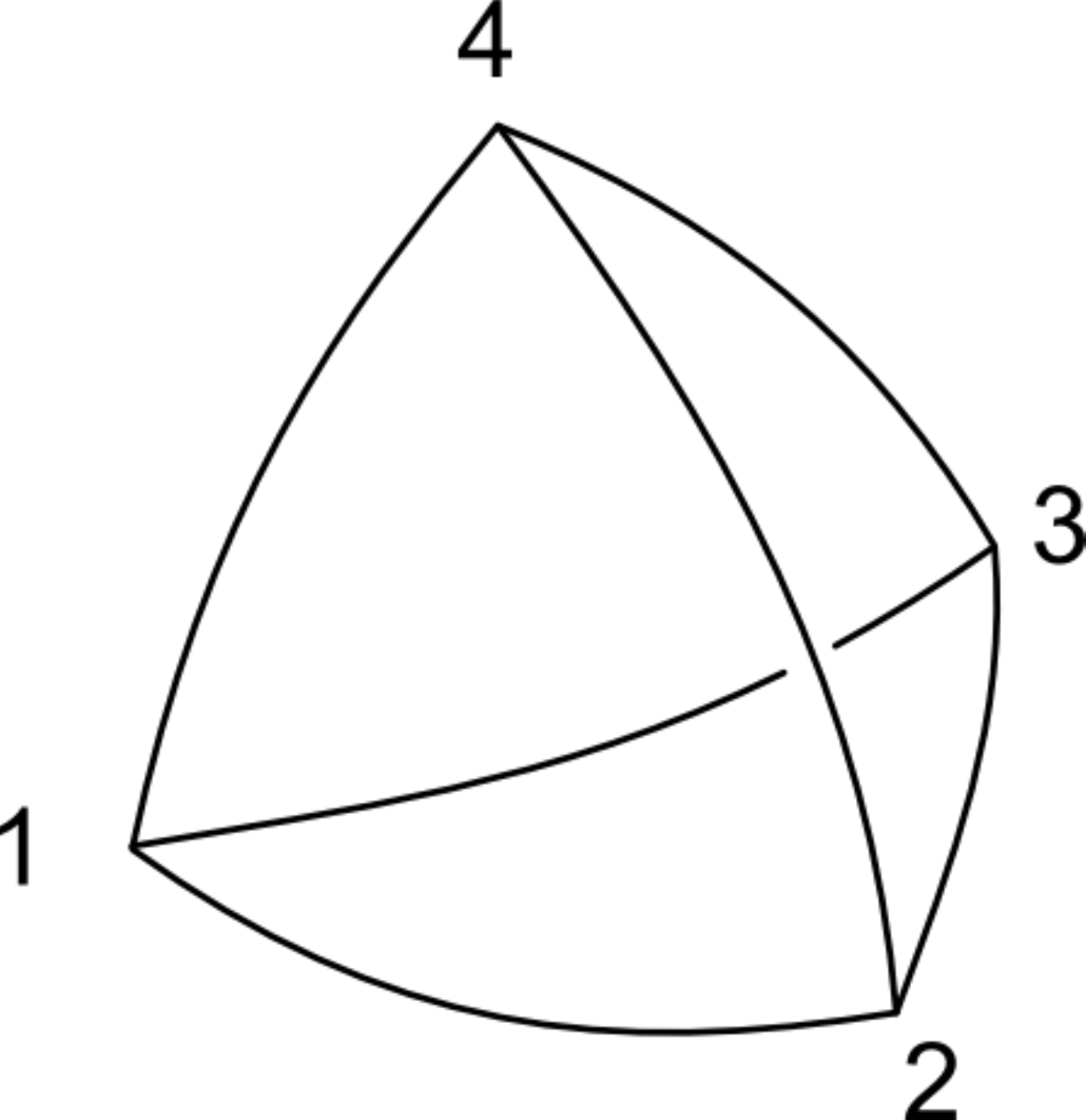}
\caption{A standard numbering of the vertices of the tetrahedron, which also induces a particular topological orientation. The BurntOrangep character of this tetrahedron is simply a device for underlining its spherical nature. However, similar pictures result from stereographic projection of a spherical tetrahedron in $\rmS^3$ onto $\R^3$. This projection sends great two-spheres of $\rmS^3$ into spheres of $\R^3$, though possibly with different radii. Note that the convex or concave aspect of the stereographically projected spherical tetrahedron has no intrinsic meaning.}%
\label{fig_topordering}
\end{center}%
\end{figure}%
Each face is traversed in a  counterclockwise sense when seen from the outside of the tetrahedron.
%\footnote{This is a well defined concept thanks to the convexity requirement.} 
This is consistent with the tetrahedron's topological orientation. The normals appearing in the holonomies, Eq. \eqref{eq_O}, are hence the outward pointing normals to the face whenever the base point $P$ of the holonomy $O(P)$ lies on that face (right-handed convention). There is no natural common base point for all four faces. However, any three faces do share a point. Pick faces $\ell=1,2,3$, which share vertex 4, and base the holonomies at this vertex: 
\be
O_\ell := O_\ell (4).
\ee
Then, in the case of holonomies $O_{1,2,3}(4)$, the vectors $\{\hat n_1(4),\hat n_2(4),\hat n_3(4)\}$ are outward normals to their respective faces in the frame of vertex 4. Clearly, this is not the case for the normal $\hat n_4(4)$. Thus, we must specify the path used to define the holonomy around face 4 and its transport to vertex 4. By now, this path is completely fixed by the curved closure equation. It consists of defining $O_4(2)$ in an analogous way to the $O_{1,2,3}(4)$ and then parallel transporting it to vertex 4 through the edge $(42)$. The set of relevant paths is shown in \autoref{fig_paths}. Up to the choice of the base point, this is manifestly the simplest (and shortest) set of paths going around each face in the order required by the closure equation and composing to the trivial loop. For this reason, we will call these simple paths. %(As a side remark, notice how an analogous ``simple'' choice of path is problematic for a general polyhedron.) %
\begin{figure}%
\begin{center}%
\includegraphics[height=3cm]{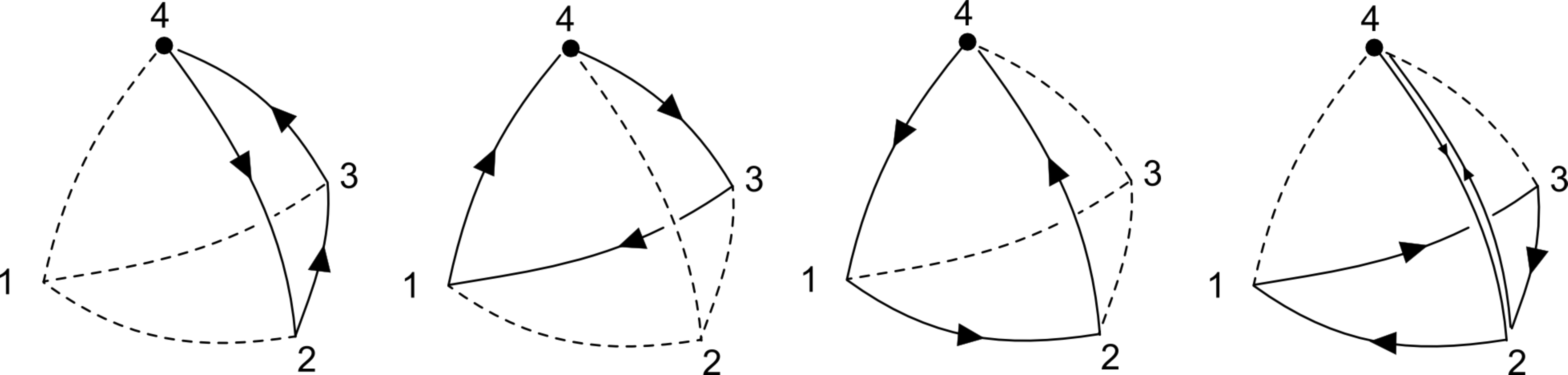}
\caption{The set of simple paths used to interpret the holonomies $\{O_\ell\}$.}%
\label{fig_paths}
\end{center}%
\end{figure}%

The holonomies along the simple paths, $\{O_\ell\}$, can be expressed more explicitly by introducing the edge holonomies $\{o_{m\ell}\}$, encoding the parallel transport from vertex $\ell$ to vertex $m$ along the edge connecting them (we use leftward composition of holonomies). Thus, $o_{\ell m} \equiv o_{m\ell}^{-1}$ , and
\be
\left\{
\begin{array}{l}
O_1 = o_{43}o_{32}o_{24}\\
O_2 = o_{41}o_{13}o_{34}\\
O_3 = o_{42}o_{21}o_{14}\\
O_4 = o_{42} O_4(2) o_{24} = o_{42}\left[ o_{23}o_{31}o_{12}\right] o_{24}
\end{array}
\right.\,.
\ee
Let us stress once more that, since the closure equation is preserved by a cyclic permutation of the holonomies, the assignment to a specific holonomy of the label ``4'' is indeed an extra input needed by the reconstruction. We call this vertex the special vertex.

Another important symmetry of the closure equation is its invariance under conjugation of the four holonomies by a common element of $\SO$:
\be
O_\ell \mapsto R O_\ell R^{-1}\,,\quad\text{with } R\in\SO\,. 
\ee
This maps the areas $a_\ell$ into themselves, and the normals $\hat n_\ell$ into $ \bR \hat n_\ell$.  %
%{\bf [ or is it $\bR^{-1}\hat n$??] \textcolor{red}{I think it is correct as it is, double check please} add comment about bold face notation} %
(Here and in the rest of the paper, bold-face symbols stand for matrices in the fundamental representation; e.g. in the previous equation, $\bR$ is the $3\times3$ matrix corresponding to $R\in \mathrm{SO}(3)$.) This symmetry can be interpreted either as a change of reference frame at the base point 4, or as the effect of a further parallel transportation of the $\{O_\ell\}$ along another piece of path from vertex 4 to some other base point. The latter interpretation is particularly compelling when $R=o_{24}$ : the result of this transformation is an exchange of the r{\^o}le of vertices (and therefore faces) 4 and 2. We conclude that picking vertex 4 or 2 as special, are gauge equivalent choices. So, it is more appropriate to refer to the edge $(24)$ as the special edge rather than referring to 2 or 4 as special vertices. 

%Since it is possible to rotate the reference frame at all vertices, the formulation which uses the edge holonomies is highly redundant. Indeed, it one can always gauge transform them according to
%\be
%o_{m \ell} \mapsto R_m o_{m\ell} R_\ell^{-1}\,.
%\ee
%Since this symmetry has the physical meaning of local invariance under rotations, it is interpreted as a gauge symmetry.  In principle, it is possible to work with certain holonomies gauge fixed according to some pattern of choice. Nonetheless, we will not pursue this strategy, and work only with face holonomies defined up to a global rotation.\\

%\textcolor{BurntOrange}{
A set of holonomies  that close $\{O_\ell \; | \; \prod_\ell O_\ell =\E \}$ modulo simultaneous conjugation, is naturally interpreted as the moduli space of $\mathrm{SO}(3)$ flat connections on a sphere with four punctures. Indeed, since the holonomies of a flat connection can only depend on the homotopy class of the (closed) path along which they are calculated, these connections are maps from the fundamental group of the punctured sphere to $\mathrm{SO}(3)$. Therefore, the moduli space of flat connections is this space of maps modulo conjugation: 
\begin{align}
\mathcal{M}_\text{flat} \big[L\text{-punctured } \rmS^2  , \mathrm{SO}(3)\big] & \cong  \mathrm{Hom}\big[\pi_1( L\text{-punctured } \rmS^2  ), \mathrm{SO}(3)\big]\big/\mathrm{SO}(3) \label{eq_isomorph} \notag\\
& \cong \big\{ O_1, \dots , O_L \in \mathrm{SO}(3) \; | \; O_L\cdots O_1 = \mathrm e \big \} \big/ \text{conjugation} 
\end{align}
Conjugation is the residual gauge freedom left at the arbitrarily chosen base point of the holonomies. The connection with the tetrahedron's geometry arises from the observation that the fundamental group of the 4-punctured sphere is isomorphic to that of the tetrahedron's one-skeleton. %:
%\be
%\pi_1( 4\text{-punctured } \rmS^2  ) \cong \pi_1( \text{tetrahedron's 1-skeleton}  ).
%\ee
However, this isomorphism is not canonical, and constitutes the extra piece of information that is needed to run the reconstruction, i.e. the knowledge of the precise paths associated to the $\{O_\ell\}$.
%}

The choice of a special edge, e.g. (24), breaks permutation symmetry. However, conjugation is a true symmetry of the problem, and is the analogue of rotational invariance for the flat case. Therefore, any quantity with an intrinsic geometrical meaning must be obtained through conjugation invariant combinations of the $\{O_\ell\}$. The normals $\{\hat n_\ell\}$ are not gauge invariant observables, but their scalar and triple products are.

Scalar products between the normals have a clear meaning: they encode the dihedral angles between the faces of the tetrahedron. Because the faces of the tetrahedron are flatly embedded, these dot products are invariant along the edge shared by two faces and hence the dihedral angles are well defined.  For faces 1, 2, and 3, the situation is simple. The holonomies $\{O_\ell\}$ and the normals appearing in their exponents are defined at vertex 4, which is shared by all three faces. Therefore, indicating with $\theta_{\ell m}$ the (external) dihedral angle between faces $\ell$ and $m$, see \autoref{fig_dihedral},%
\begin{figure}%
\begin{center}%
\includegraphics[height=3.1cm]{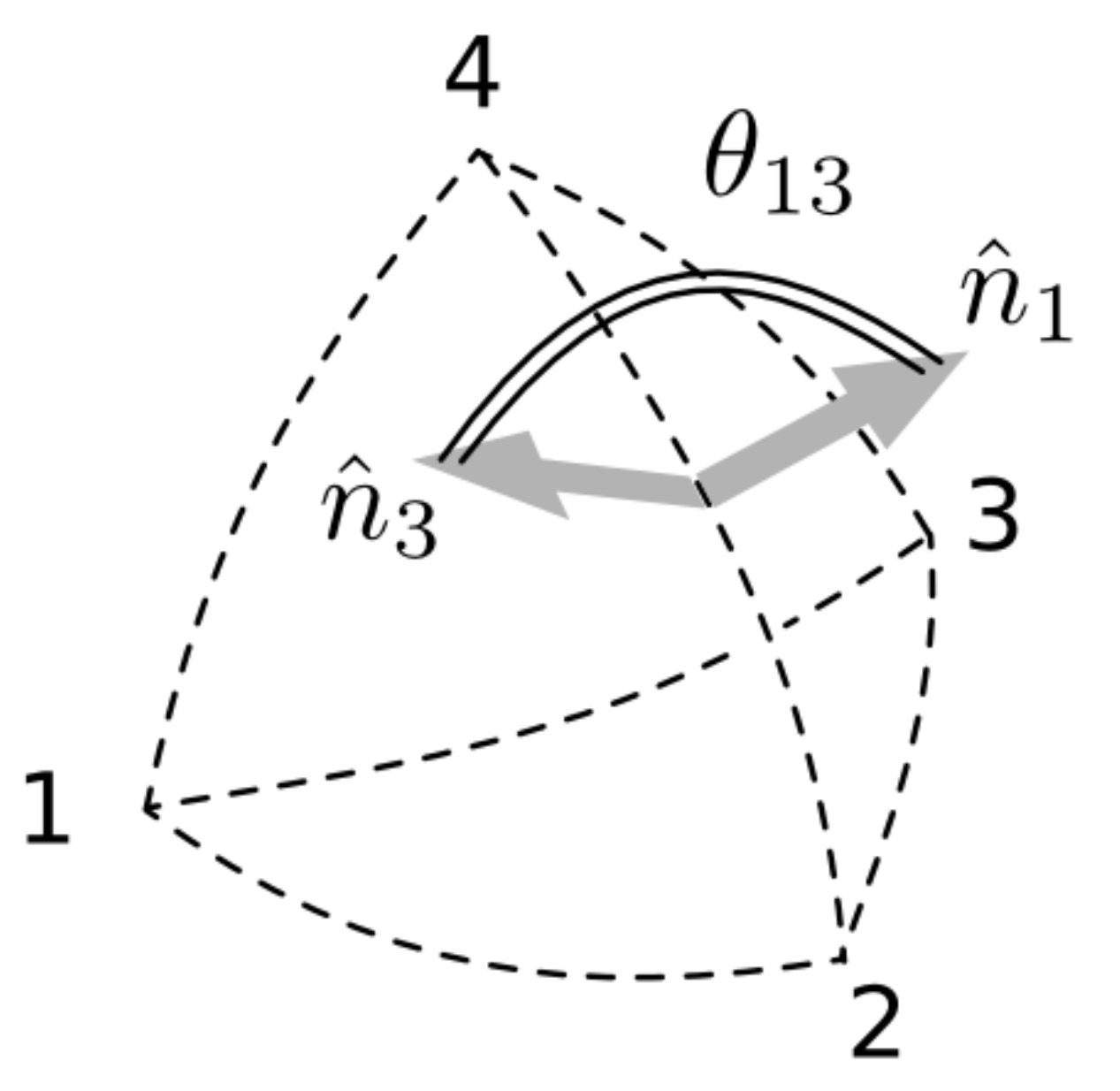}
\caption{The dihedral angle $\theta_{\ell m}$ spans the arc from outward normal $\ell$ to $m$. Here we illustrate the case $\{\ell,m\}=\{1,3\}$.}%
\label{fig_dihedral}
\end{center}%
\end{figure}%
\be
\cos \theta_{\ell m} = \hat n_\ell \cdot \hat n_m\,,\quad\text{for }\ell,m\in\{1,2,3\}\,.
\ee
Recall that $O_4$ is first defined at vertex 2 and then parallel transported to vertex 4 along the edge $(24)$. Because of the gauge equivalence of 2 and 4 as special vertices, and because vertex 2 is shared by faces 1, 3, and 4, calculating the dihedral angles between these face is as simple as before:
\be
\cos \theta_{\ell m} = \hat n_\ell \cdot \hat n_m\,,\quad\text{for }\ell,m\in\{1,3,4\}\,.
\ee
To see this in a more direct way, note that, for example, $\cos \theta_{14} = \hat n_1(2) \cdot \hat n_4(2) = \left[\textbf{o}_{24} \hat n_1(4)\right]\cdot \left[\textbf{o}_{24} \hat n_4(4)\right] =  \hat n_1(4) \cdot \hat n_4(4)$, which is exactly the result of the previous equation. 

The remaining dihedral angle, between the opposite, special faces 2 and 4, is more delicate. This is because neither of the vertices 2 or 4 is shared by the faces 2 or 4. To calculate $\cos\theta_{24}$, we use the normals at vertex 3:
\begin{subequations}
\begin{align}
\cos \theta_{24} &= \hat n_2(3) \cdot \hat n_4(3) \notag\\
&= \left[ \textbf{o}_{34}\hat n_2(4) \right] \cdot \left[ \textbf{o}_{32}\textbf{o}_{24}\hat n_4(4) \right] \notag\\
& = \hat n_2(4) \cdot \bO_1 \hat n_4(4)\,.
\end{align}
The paths used for transporting the normals from vertex 3 to vertex 4 are not accidental; they lie within their own face up to the point where the face holonomy is based, and then move on, when necessary, to vertex 4 through the special edge (24). All paths lying within a single face are equivalent because of the flat embedding and so we use the most convenient choice. 

Had we chosen to define $\theta_{24}$ at vertex 1 instead of 3, the result would have been
\be
\cos \theta_{24} = \hat n_2(4) \cdot \bO_3^{-1} \hat n_4(4)\,.
\ee
\end{subequations}
A quick check shows that these two results are equivalent, thanks to the closure equation and the relation $\bO_\ell \hat n_\ell = \hat n_\ell$. Summarizing,
\be
\left\{
\begin{array}{ll}
\cos \theta_{24} = \hat n_2 \cdot \bO_1 \hat n_4 = \hat n_2 \cdot \bO_3^{-1} \hat n_4 \\
\cos \theta_{\ell m} = \hat n_\ell \cdot \hat n_m & \text{for } \{\ell,m\} \neq \{2,4\}
\end{array}
\right.
\label{eq_dihedral}
\ee
Notice that $\theta_{\ell m}\in(0,\pi)$ in order to have a convex tetrahedron, and although this condition would be redundant for a tetrahedron in flat space,
%\footnote{However, convexity \textit{is} a crucial ingredient for the uniqueness part of Minkowski's theorem for general flat polyhedra.} 
one could use the sphere's non-trivial topology to build non-convex spherical tetrahedra.\footnote{To construct an example, one can replace one of the edges of a standard convex spherical tetrahedron with its complement with respect to the great circle it lies on. Another example can be constructed by replacing a whole face with its spherical complement. \label{fn1}} We are not interested in reconstructing such objects. Moreover, the previous condition implies that we can invert Eq. \eqref{eq_dihedral} to obtain the values of the $\{\theta_{\ell m}\}$ themselves. These formulas require only data entering the curved closure equation, and not the edge holonomies $\{o_{\ell m}\}$, as expected from considerations of gauge invariance.

There is still a subtle point to clarify. How can the directions of the outward normals $\{\hat n_\ell\}$ be extracted from the $\{O_\ell\}$? The face areas of the tetrahedron are positive real numbers $a_\ell$ lying in the interval $(0,2\pi)$ due to the tetrahedron's convexity (see footnote \ref{fn1}). However, the holonomy $O_\ell$ cannot distinguish between two triangles lying on the same great two-sphere in $\rmS^3$ that have areas $a$ and $(2\pi-a)$, respectively, and corresponding normals $\hat n$ and $-\hat n$. In formulas:
\be
\exp \left\{a \ \hat n \cdot \vec J \;\right\}= \exp \left\{ (2\pi-a) (-\hat n)\cdot \vec J \;\right\}\,.
\label{eq_ambiguity}
\ee
This is a consequence of the fact that both the trivial loop and a great circle have trivial $\SO$-holonomy. To resolve this ambiguity, it is enough to appeal to convexity by checking the signs of the triple products among the normals. Indeed, the triple products are naturally associated to the vertices of the tetrahedron, and their signs relate to its convexity (as well as to our choice of its topological ordering, and to the outward pointing property of the normals $\{\vec n_\ell\}$), see \autoref{fig_triple}. %
\begin{figure}%
\begin{center}%
\includegraphics[height=3.3cm]{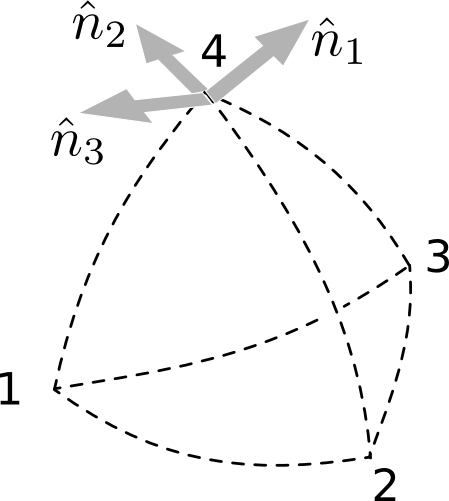}%
\caption{The three vectors involved in the triple product at vertex 4. Given the topological orientation of the tetrahedron, its convexity, and supposing all normals are outward pointing, one finds that $\sgn\left[\left( \hat n_1 \times \hat n_2\right) . \hat n_3\right]>0$.}%
\label{fig_triple}%
\end{center}%
\end{figure}%

Concretely, this translates into the following requirements for the normals:
\be
\left\{
\begin{array}{ll}
\text{at vertex 4: } & \left[ \hat n_1(4) \times \hat n_2(4) \right]\cdot \hat n_3(4) >0\\
\text{at vertex 2: } & \left[ \hat n_1(2) \times \hat n_3(2) \right]\cdot \hat n_4(2) >0\\
\text{at vertex 1: } & \left[ \hat n_2(3) \times \hat n_1(3) \right]\cdot \hat n_4(3) >0\\
\text{at vertex 3: } & \left[ \hat n_3(1) \times \hat n_2(1) \right]\cdot \hat n_4(1) >0
\end{array}
\right.\,.
\label{eq_tripleproducts}
\ee
After parallel transporting to the common base point, vertex 4, these conditions read
\be
\left\{
\begin{array}{ll}
\left( \hat n_1 \times \hat n_2 \right)\cdot \hat n_3 &>0\\
\left( \hat n_1 \times \hat n_3 \right)\cdot \hat n_4 &>0\\
\left( \hat n_2 \times \hat n_1 \right)\cdot \bO_1 \hat n_4& >0\\
\left( \hat n_3 \times \hat n_2 \right)\cdot \bO_3^{-1} \hat n_4& >0
\end{array}
\right.\,.
\label{eq_triple}
\ee
A moment of reflection shows that these conditions are exactly what is needed to solve the ambiguity expressed in equation \eqref{eq_ambiguity}. In fact, among the $2^4$ possible redefinitions of the normals by change of signs $\{\hat n_\ell \}\mapsto \{\pm_\ell \hat n_\ell\}$, one and only one of them satisfies Eq. \eqref{eq_triple}.

It is interesting to express the intrinsic geometrical quantities of the tetrahedron, such as areas, dihedral angles, and triple products, directly in terms of the holonomies $\{O_\ell\}$. The simplest conjugation invariant set of observables are traces of products of the $\{\bO_\ell\}$. These turn out to be quite involved. A convenient alternative is given by the same invariants for the \textit{lifts} of the $\{O_\ell\}$ to $\SU$. Call these lifts $\{H_\ell\}$, and their matrices in the fundamental representation $\{\bH_\ell\}$. The twofold ambiguity associated with the lift reflects the geometric ambiguity of Eq. \eqref{eq_ambiguity}, which is already present at the level of $\SO$. It is tempting to conjecture that considering $\SU$ closures solves this ambiguity, and that the $\SU$ holonomies can be automatically associated to the spin connection of the homogeneously curved space. Unfortunately, this is not the case, since by multiplying the geometrical values of any two (or four) of the $\SU$ holonomies by $-\mathbf{1}$, one obtains another sensible closure equation that looses its geometrical interpretation.\footnote{A more sophisticated attempt to make this work would consist in allowing non-convex tetrahedra. Indeed, taking the equatorial complement of one side of  a standard tetrahedron would modify the area of the two adjacent faces from $a_\ell$ to $2\pi-a_\ell$ at the price of obtaining a non-convex tetrahedron. The problem with this extension is that there is no unique choice of sides to complement. Hence the uniqueness of the reconstructed geometry would be lost.} Hence, we are lead to allow any consistent lift with the $\SU$ closure 
\be
H_4H_3H_2H_1=\E,
\ee
and eventually correct for the sign of (an even number of) the holonomies in such a way that all the inequalities of Eq. \eqref{eq_triple} are satisfied. A slightly different way of stating this, with closure only holding up to a sign, is that what we are really considering are $\mathrm{PSU}(2)$ closures, and only these are in one to one correspondence with curved tetrahedra. In the following we will mostly deal with $\SU$ holonomies, to which we associate geometries in an almost one-to-one way.

The convenience of using the $\{\bH_\ell\}$ comes from the simple identity: %{\bf do I want to put an extra minus at the exponent to get the matching with formulas for SO(3)??}
\be
\bH = \exp\left\{ a \hat n \cdot \vec \btau \,\right\}= \cos{\frac{a}{2}}\;\mathbf{1} - \I \sin \frac{a}{2} \;\hat n \cdot \vec{\bsigma},
\ee
where $\vec\sigma$ are the Pauli matrices, and $\vec\tau:=-\frac{\I}{2}\vec\sigma$. %
%
%%The simplest conjugation invariant set of observables one can think of are the traces of products of the $\{O_\ell\}$. {\bf is it a complete set? in this case it was obvious that it was possible... though expression are nice, and this is worth mentioning} It turns out that all the geometric quantities listed above can be very naturally expressed in terms of traces of holonomies. To show this fact the following two formulas are sufficient:
%\be
%\bO = \exp a \hat n.\vec{\bJ} = \mathbf{1} + (\cos(a)-1) \mathbf{P}^\perp_{\hat n} + \sin(a) \hat n . \vec{\bJ}\,,
%\ee
%where $\mathbf{P}^\perp_{\hat n} = \left( \mathbf{1} - \hat n \otimes \hat n^T \right)$ is the projector on the two dimensional space orthogonal to $\hat n$, and
%\be
%\frac{1}{2}\Tr\left(\bJ^i \bJ^j\right) =- \delta^{ij}\,.
%\ee
Define the connected part of the half-trace of the product of $p$ holonomies, $\langle\; \stackrel{p}{\overbrace{H_\ell \cdots H_m }} \; \rangle_C \;$:
\begin{subequations}
\begin{align}
\langle H \rangle_C & := \frac{1}{2}\Tr(\bH),\\
\langle H_\ell H_m \rangle_C  & := \frac{1}{2}\Tr(\bH_\ell \bH_m) - \frac{1}{4}\Tr(\bH_\ell)\Tr(\bH_m),\\
\langle H_\ell H_m H_q \rangle_C  & := \frac{1}{2}\Tr(\bH_\ell \bH_m \bH_q) - \left[ \frac{1}{4}\Tr(\bH_\ell)\Tr(\bH_m\bH_q) + \text{cyclic} \right] + \frac{1}{4}\Tr(\bH_\ell)\Tr(\bH_m)\Tr(\bH_q),\\
\text{etc.}\notag
\end{align}
\end{subequations}
It is then straightforward to check that the geometrical quantities of interest are normalized versions of these quantities:
\begin{subequations}
\begin{align}
\cos \frac{a_\ell}{2} & = \pm_\ell \langle \bH_\ell \rangle_C,\\ 
\cos \theta_{\ell m} & = \hat n_\ell . \hat n_m = - \frac{\pm_\ell \pm_m\langle \bH_\ell \bH_m \rangle_C}{\sqrt{1-\langle \bH_\ell \rangle_C^2}\sqrt{1-\langle \bH_m \rangle_C^2}} \quad\text{for } \{\ell,m\}\neq\{2,4\}  \label{eq_scalarOO},\\
\left( \hat n_\ell \times \hat n_m \right).\hat n_q & = - \frac{\pm_\ell \pm_m \pm_q \langle \bH_\ell \bH_m \bH_q \rangle_C}{\sqrt{1-\langle \bH_\ell \rangle_C^2}\sqrt{1-\langle \bH_m \rangle_C^2}\sqrt{1-\langle \bH_q \rangle_C^2}} \quad\text{for } \{\ell,m,q\}=\{1,2,3\} \text{ or } \{1,3,4\} ,
\label{eq_tripleOOO}
\end{align}
\end{subequations}
with the appropriate generalization for $\cos\theta_{24}$ and the missing triple products. The $\{\pm_\ell\}$ signs can be thought of as representing the branches of the respective square roots, and are uniquely fixed by imposing the positivity of the triple products, i.e. the convexity of the tetrahedron. They are eventually related to the areas via $\pm_\ell = \sgn\sin a_\ell$, which is another way of stating the ambiguity $\{ a_\ell , \hat n_\ell\} \mapsto \{2\pi - a_\ell,-\hat n_\ell\}$.%, as it is clear from the first of these equations.

At this point we are left with the simple exercise of reconstructing a spherical tetrahedron from its known dihedral angles $\{\cos\theta_{\ell m}\}$. Notice that the areas $\{a_\ell\}$ are not needed. Their consistency with respect to the reconstructed geometry will be proved in complete generality in \autoref{sec_theorem}. The key equation in the reconstruction is the spherical law of cosines, relating the edge lengths of a spherical triangle to its face angles. With the notation of \autoref{fig_sphcosinelaw}, this law and its inverse read%
\begin{figure}%
\begin{center}%
\includegraphics[height=3cm]{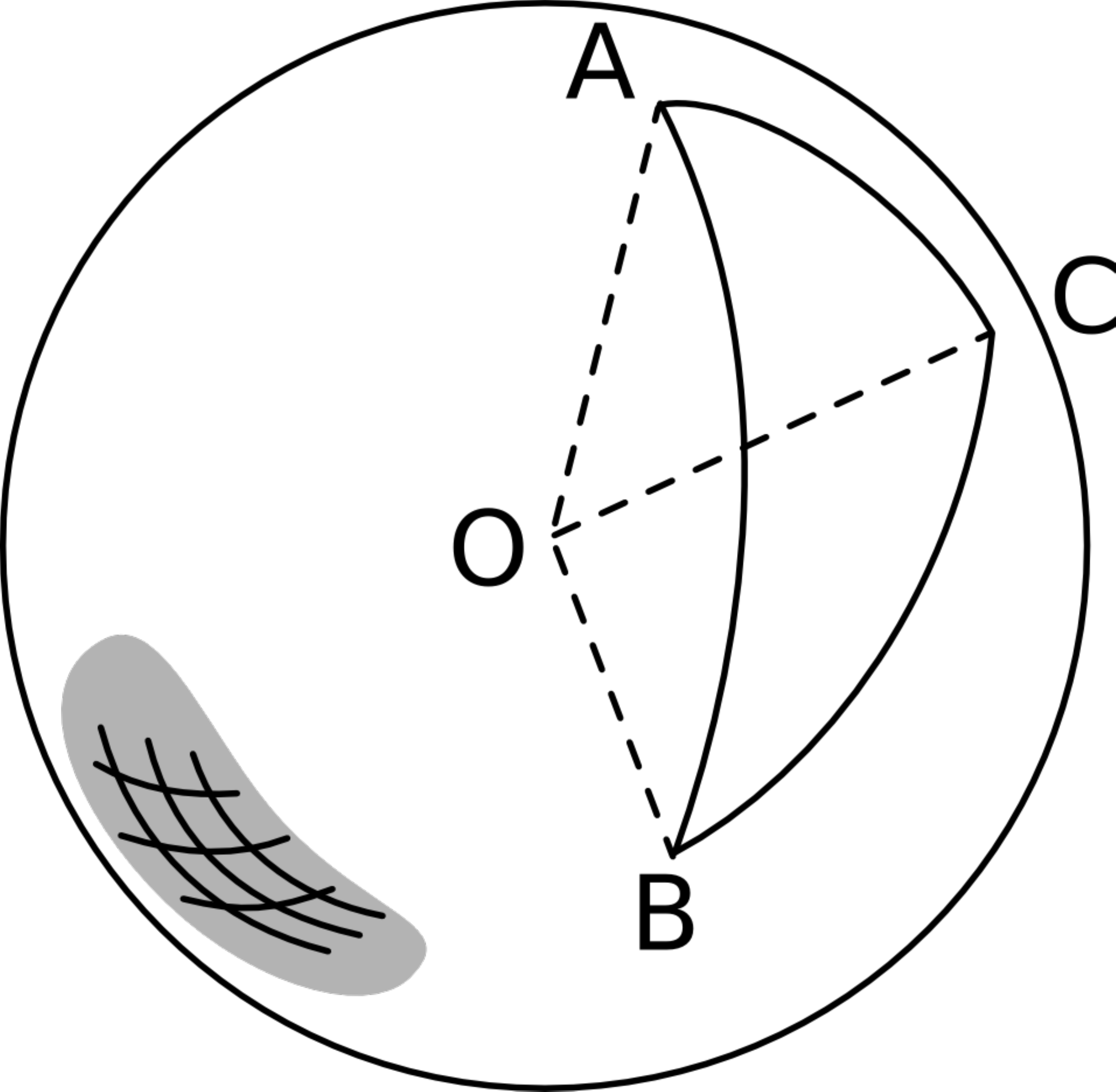}
\caption{A spherical triangle illustrating the notation for Eqs. \eqref{eq_sphcos} and \eqref{eq_sphcosinv}, the spherical cosines laws.}%
\label{fig_sphcosinelaw}
\end{center}%
\end{figure}%
\begin{subequations}
\begin{align}
\cos \widehat{C} & = \frac{\cos \wideparen{AB} - \cos\wideparen{AC} \cos\wideparen{BC}}{\sin\wideparen{AC} \sin\wideparen{BC}}\,,%
\label{eq_sphcos}\\
{}\notag\\
\cos \wideparen{AB} & = \frac{\cos \widehat C+ \cos \widehat A \cos \widehat B}{\sin \widehat{A} \sin\widehat{B}}\,,%
\label{eq_sphcosinv}
\end{align}
\end{subequations}
where $\wideparen{AB}$ is the arclength (on the unit sphere) between the vertices $A$ and $B$, and $\widehat{C}$ is the angle between the arcs $AC$ and $BC$ at point $C$. By putting an infinitesimal sphere around the vertex $\ell$ of the spherical tetrahedron, and looking at the spherical triangle defined by the intersections of this sphere with the edges stemming from vertex $\ell$, one can use Eq. \eqref{eq_sphcos} to deduce the three face angles at the vertex $\ell$ from the tetrahedron's dihedral angles. Once all the face angles are known, Eq.  \eqref{eq_sphcosinv} yields the edge lengths for each face of the tetrahedron. Therefore, by using just one formula and its inverse, it is possible to deduce the full geometry of the spherical tetrahedron from its dihedral angles. This is possible, in the curved case, because the radius of curvature provides a natural scale to translate angles into arclengths. In this respect, the flat case is a degenerate limit in which scale invariance appears. In the flat closure,  Eq. \eqref{eq_flatclosure}, the areas can all be rescaled by a common factor without altering the normals. No analogously simple symmetry is present in the curved case.

%As a concluding remark, notice that the following transformation gives rise to a new closure equation which reconstructs the same geometrical tetrahedron, but with a different assignment of names to its vertices. It is as if the ``special edge'' was moved from edge (24) to edge (13). The transformation is
%\be
%(O_4,O_3,O_2,O_1)\mapsto(O_1, O_3, O_3^{-1} O_4 O_3,O_2)\,.
%\ee
%In a sense, this is the geometrically meaningful way  of permuting the r{\^o}le of the vertices within the tetrahedron.\footnote{We thank Wojciech Kami\'nski for this observation.} {\bf \color{red} special case of more general transformation changing the presentation of the first homotopy group??}\\

%-----------------------------------------------------------------------------------------
\section{A first look at the hyperbolic case}\label{sec_hypgeom}

In the case of hyperbolic tetrahedra, the reconstruction theorem proceeds in essentially the same way as above. Once again, the faces of the curved tetrahedron are required to be flatly embedded, which implies the holonomies around them have a form completely analogous to those in the spherical case (Eq. \eqref{eq_O}):
\begin{align}
O_\ell (P) = \exp \left\{-a_\ell \hat n_\ell(P).\vec J\;\right\}\,,
\label{eq_Ohyp}
\end{align}
where all the symbols are interpreted in the same manner, and the minus sign is due to the negative sign of the curvature. A crucial fact about this formula is that the holonomies are again in $\SO$, and not in some other group with different signature. The simple reason for this is that $\SO$ is the group of symmetries of the tangent space (at a point) of both $\rmS^3$ and $\rmH^3$.

 We deduce the dihedral angles of the tetrahedron following similar reasoning to that of the previous section. The extra minus sign of Eq. \eqref{eq_Ohyp} has consequences only for the formulas that express the triple products of the normals in terms of connected traces (Eqs. \eqref{eq_tripleproducts} and \eqref{eq_tripleOOO}); the right-hand sides of these equations should be multiplied by -1. The formulas for the dihedral angles, which involve two normals, are only sensitive to the overall agreement in sign of the triple products, which is granted in both the spherical and hyperbolic cases. This latter fact will be crucial in the following.

Once the dihedral angles have been calculated, the tetrahedron can be straightforwardly reconstructed using the hyperbolic law of cosines:
\begin{subequations}
\begin{align}
\cos \widehat{C} & = - \frac{\cosh \wideparen{AB} - \cosh\wideparen{AC} \cosh\wideparen{BC}}{\sinh\wideparen{AC} \sinh\wideparen{BC}}\,,%
\label{eq_hypcos}\\
{}\notag\\
\cosh \wideparen{AB} & = \frac{\cos \widehat C+ \cos \widehat A \cos \widehat B}{\sin \widehat{A} \sin\widehat{B}}\,.%
\label{eq_hypcosinv}
\end{align}
\label{eq_hypcoslaw}
\end{subequations}
Note the extra minus sign in the first equation. These formulas conclude the list of ingredients needed for the reconstruction in the hyperbolic case.

In section \ref{sec_twosheeted}, however, we shall see that these ingredients are not quite enough to cover all the possible hyperbolic cases naturally arising from the closure equation. A new generalization of hyperbolic geometry has to be introduced.

%That would be it for ``usual'' hyperbolic tetrahedra, i.e. tetrahedra living on a single hyperbolic sheet. However, while such tetrahedra have faces whose areas are in the range $(0,\pi)$, it turns out that we need to make sense of hyperbolic geometries for holonomies whose trace class is $(0,2\pi)$. To take care of these cases, we introduce -- to our knowledge for the first time -- a meaningful extension of hyperbolic polygons encompassing polygons extending across both sheets of a two-sheeted hyperboloid. Before introducing this unusual geometries, and describing their properties, we spend a few words explaining what exactly pushes us to consider them. \\

%-----------------------------------------------------------------------------------------
\section{Spherical or hyperbolic? The Gram matrix criterion}\label{sec_gram}

Up to now we have described two possible reconstruction procedures, one for spherical and one for hyperbolic tetrahedra. Nonetheless, the starting point we are proposing for the reconstruction theorem is the same closure equation, Eq. \eqref{eq_closure}. The natural question arises, whether there is an \textit{a priori} criterion to decide which type of tetrahedron one should reconstruct when given only the four closing holonomies (and the choice of a special edge). There is such a criterion. The key is the unambiguousness character of the dihedral angles discussed above: once a sign (either for the moment) of the four triple products has been fixed, the dihedral angles of the curved tetrahedron are uniquely determined, irregardless of the curvature. But also, the dihedral angles encode all the necessary information to reconstruct the full tetrahedron, including its curvature. In this section we briefly review how the curvature can be deduced from the dihedral angles alone. While part of this is standard, it allows us to introduce concepts and notation useful in the following section. 

%Finally, notice that once the curvature is known, the previous two sections show how, by applying repeatedly the spherical and hyperbolic laws of cosines, one can actually reconstruct uniquely the tetrahedron's geoemtry, provided it exists.\footnote{If the tetrahedron geoemtry turns out to be flat, the knowledge of its dihedral angles only does not allow to reconstruct its edge lengths explicitly. These are determined only up to a global scale. We will comment on the flat case further on.} The existence part of the theorem will be discussed later.\\

To begin, we reverse the logic, and suppose we are actually given a tetrahedron, flatly embedded in a space of constant positive, negative or null curvature. Then, define its Gram matrix, as the matrix of cosines of its (external) dihedral angles:
\begin{align}
\Gr _{\ell m} := \cos \theta_{\ell m} \quad\text{for } \ell\neq m, \quad\text{and}\quad \Gr_{\ell\ell} := 1 \quad \forall \ell\,.
\end{align}
One of the main properties of the Gram matrix is that the sign of its determinant reflects the spherical, hyperbolic, or flat nature of the tetrahedron:
\begin{align}
\sgn \det \Gr =
\left\{ \begin{tabular}{rl}
-1 & \text{if the tetrahedron is hyperbolic}\\
0 & \text{if it is flat}\\
+1 & \text{if it is spherical}
\end{tabular}\right. .
\end{align}
A straightforward way to understand this result is by embedding in $\mathbb R^4$. 

Let us start from the flat case, the simplest one. In this case, the parallel transport is trivial and $\Gr^\text{(flat)}_{\ell m} = \hat n_\ell \cdot \hat n_m = \delta_{ij} n^i_\ell n^j_m$. By introducing the four 4-vectors $N_\ell = (0,\hat n_\ell)$, one can write the $\Gr$ matrix in terms of the $4\times4$ matrix $N$ whose components are $N^\mu_{\phantom\mu\ell}$, we denote the component index with greek letters ranging from 1 to 4, and:
\begin{align}
\det \Gr^\text{(flat)} \equiv \det N^T N = (\det N)^2 = 0.
\end{align}
The last equality follows from the obvious fact that the $N_\ell$ are not linearly independent, since they are just four 3-vectors in disguise.  Nonetheless, these 4-vectors have a useful geometric interpretation; imagine the tetrahedron as embedded in the model space $\rmE_U^3 \subset \mathbb R^4$ orthogonal to the 4-vector $U:=(1,0,0,0)$. Then, each $N_\ell$ is the 4-normal to another hyper-plane in $\mathbb R^4$ that picks out a face of the tetrahedron when it intersects $\rmE_U^3$ orthogonally. This is depicted in one lower dimension in \autoref{fig_triangles}, where it is also clear that the (cosine of the) hyper-dihedral angle $\delta_{\mu\nu} N_\ell^\mu N_m^\nu$ is equal to the (cosine of the) tetrahedron's dihedral angle $\theta_{\ell m}$. Note that in this case, using the Euclidean ($\delta_{\mu\nu}$) or Lorentzian metric ($\eta_{\mu\nu}=\mathrm{diag}(-1,1,1,1)$) does not make any difference, since $N^0_\ell =0$. This reflects the fact that the flat case is a degenerate version of both the spherical and hyperbolic geometries. %\textcolor{BurntOrange}{\bf Double check whether now this makes sense}

\begin{figure}%
\begin{center}%
\includegraphics[width=0.9\textwidth]{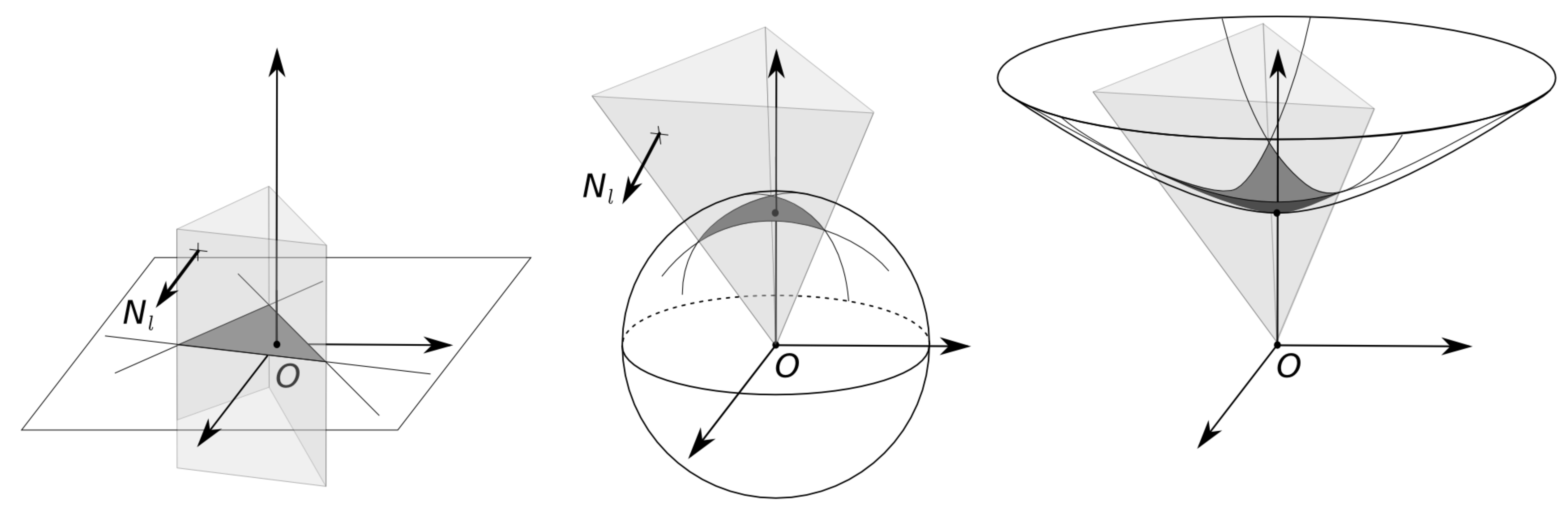}\\ \vspace{.5cm}
\includegraphics[height=3cm]{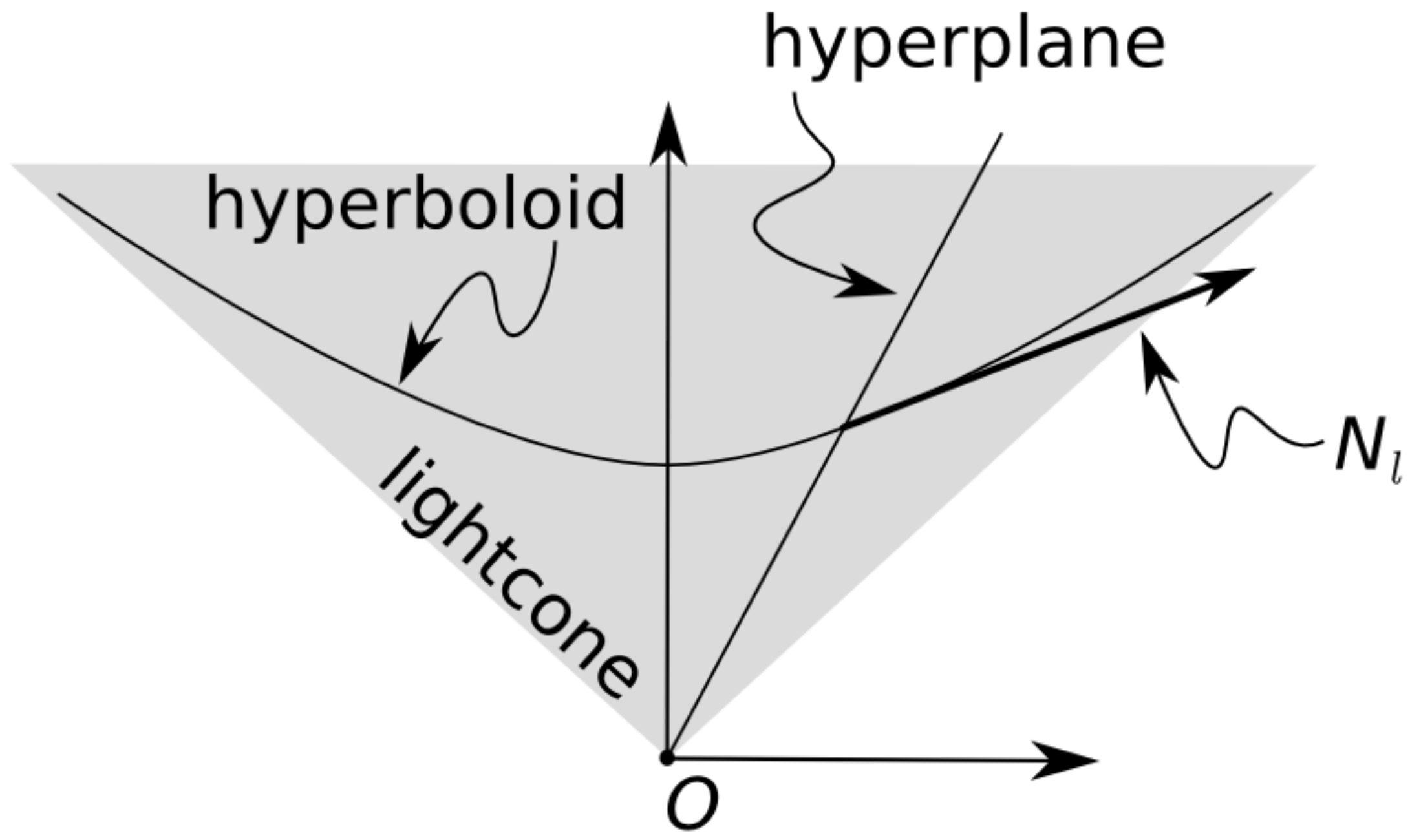}
\caption{A one-dimension lower representation of a flat, a spherical, and a hyperbolic tetrahedron as embedded in $\mathrm{E}^3_U\subset\mathbb{R}^4$, $\mathrm{S}^3_u\subset\mathbb{R}^4$, and $\mathrm{H}^3_u\subset\mathbb{R}^4$, respectively. The lower picture shows a section of the hyperbolic case to highlight the Lorentzian representation used for the hyperboloid.}%
\label{fig_triangles}
\end{center}%
\end{figure}%

In the spherical case, one embeds the tetrahedron into the unit sphere
%\footnote{$\rmS^3_u :=\{ r^\mu \in \mathbb R^4 | \delta_{\mu\nu}r^\mu r^\nu =1  \}.$} 
$\rmS_u^3\subset\mathbb R^4$. The $\ell$-th face of the curved tetrahedron will then lie on a great 2-sphere of $\rmS_u^3$ identified by the intersection of the unit sphere with a hyper-plane passing through the origin of $\mathbb R^4$ and orthogonal to the 4-vector $N_\ell$ (we use the same symbol as in the flat case). Once more, the (cosine of the) hyper-dihedral angle $\delta_{\mu\nu} N_\ell^\mu N_m^\nu$ is equal to the (cosine of the) dihedral angle $\theta_{\ell m}$ between the faces $\ell$ and $m$ of the tetrahedron (provided orientations are chosen consistently). This fact gives the relation
\begin{align}
\Gr_{\ell m}^\text{(sph)} = \delta_{\mu\nu} N^\mu_\ell N^\nu_m\,,
\end{align}
from which it follows
 \begin{align}
\det \Gr^\text{(sph)} = (\det N)^2 >0 , 
\end{align}
where the zero value has been excluded because it would correspond to a degenerate tetrahedron, which we will not treat here. 
%\footnote{Different types of degeneracy may appear. We do not intend to classify them here. {\bf citation needed...?}}\\
  
The easiest way to understand the hyperbolic case (see \autoref{fig_triangles}), is in terms of a ``Wick rotation'' of the spherical one. One obtains
\begin{align}
\Gr_{\ell m}^\text{(hyp)} = \eta_{\mu\nu} N^\mu_\ell N^\nu_m\,,
\end{align}
from which it follows
 \begin{align}
\det \Gr^\text{(hyp)} = (\det \eta)(\det N)^2 <0 , 
\end{align}
since $\det\eta = -1$. Here too, the zero value has been excluded because it corresponds to degenerate cases. The new metric is needed because the Euclidean normals to the planes that intersect the unit hyperboloid
%\footnote{$\rmH^3_U :=\{ r^\mu \in \mathbb R^4 | \mu_{\mu\nu}r^\mu r^\nu =-1  \}$. In this context we need only consider the one-sheeted hyperboloid, $\mu_{\mu\nu}r^\mu U^\nu <0$.} 
$\rmH^3_u\subset \mathbb R^4$ are not tangent to the hyperboloid at the points of contact, and therefore the Euclidean scalar product between these normals does not reproduce the tetrahedron's Gram matrix. Related to this, there is the fact that the hyperboloid of \autoref{fig_triangles} has negative curvature only when calculated within the Lorentzian metric (time direction pointing upwards in the figure).% Lorentzian geometry is the most effective way of dealing with this situation.\\

Interestingly, there is a direct way to calculate the sign of the determinant of the Gram matrix just in terms of the holonomies $O_\ell$ and the choice of a special edge. We are free to choose the special vertex 4 of the curved tetrahedron to be located at the north p\^ole of $\rmS_u^3$ (or of $\rmH^3_u$, respectively), in which case $N_\ell = (0, \hat n_\ell)$ for $\ell \in\{ 1,2,3 \}$, where the $\hat n_\ell$ are determined up to a global sign by the procedure discussed in the previous section. The last three components of $N_4$ are then completely determined by the equations
\begin{align}
\cos\theta_{4\ell}=\Gr_{4 \ell} = g_{\mu\nu} N^\mu_4 N^\nu_\ell = \delta_{ij} N^i_4 N^j_\ell \,, \quad \text{with } \ell\neq4\,,
\end{align}
where $\cos \theta_{4\ell}$ is given by Eq. \eqref{eq_dihedral} and $g_{\mu\nu}$ can be either $\delta_{\mu\nu}$ or $\eta_{\mu\nu}$. Explicitly: 
\begin{align}
  N_4^i =  \frac{1}{(\hat n_1 \times \hat n_2)\cdot \hat n_3}\Big[
  \cos\theta_{41}\, \hat n_2\times\hat n_3 
+ \cos\theta_{42}\, \hat n_3\times\hat n_1 
+ \cos\theta_{43}\, \hat n_1\times\hat n_2   
\Big].
\end{align}
Hence, by using the condition that $N_4$ must be of unit norm, in either the Euclidean or the Lorentzian metric, it is easy to realize that the sign of the determinant of the Gram matrix is given by
\begin{align}
\sgn \det \Gr = \sgn \left( 1 - \delta_{ij} N_4^i N_4^j \;\right).
\end{align}

Now that we have been able to determine \textit{a priori} the nature of the curved tetrahedron, we can run the correct form of the reconstruction according to whether the holonomies turn out to be associated with a non-degenerate spherical ($\det\Gr>0$) or hyperbolic ($\det\Gr<0$) geometry. If $\det\Gr=0$, our equations should be interpreted as some sort of degenerate spherical or hyperbolic geometry, which we do not attempt to reconstruct. Indeed,  they cannot correspond to a flat tetrahedron, because in that case all holonomies should be trivial, irregardless of the shape of the tetrahedron!

In conclusion notice that one can attempt to reverse the logic presented here, by taking the four hyper-planes identified by the $N_\ell$ as the primitive variables, instead of the tetrahedron. Doing so, the above construction identifies in the \textit{spherical} case not one but 16 different tetrahedra on $\rmS_u^3$, with antipodal pairs congruent.\footnote{To visualize this, it is easier to think of a 2-sphere cut by three planes passing through its center: it gets subdivided into 8 triangles.} The way we have defined the Gram matrix picks out only one of these tetrahedra, the one for which all four normals induced by the $N_\ell$ are outgoing. Choosing one among these 16 tetrahedra is somewhat analogous to fixing the signs of the four triple products discussed in the previous section. Also, it is interesting to note that in the flat case this multiplicity does not appear, provided the tetrahedra ``opened up towards infinity'' are disallowed. What about the hyperbolic case? On the (one sheeted-)hyperbolid the situation is analogous to the flat case; however, by looking at the hyperboloid as a sort of analytical continuation---we do not intend to be precise about this claim---of the sphere, one might expect to find again a remnant of the 16-fold multiplicity. A moment of reflection shows that the tetrahedra crossing the equator in the spherical case are ``broken up'' into two pieces, both extending to infinity, and contained in the two seperate sheets of the two-sheeted hyperboloid. In the next section we shall see why these two-sheeted tetrahedra are of interest for the curved reconstruction theorem.

%-----------------------------------------------------------------------------------------
\section{Two-sheeted hyperbolic tetrahedra}\label{sec_twosheeted}

A well-known result, easily deduced from the Gau\ss-Bonnet theorem, is that the area of an hyperbolic triangle cannot be larger than $\pi$ and is given by  $a=\pi-\sum_{n=1}^3\alpha_n\leq\pi$, where $\alpha_n$ are the triangle's internal angles. (We use $a$ throughout for triangle areas and rely on context to distinguish the geometry as spherical, hyperbolic or Euclidean.) The bound is saturated by ideal triangles, i.e. triangles with vertices ``at infinity''.
%\footnote{This fact is more-or-less self evident from the Poincar\'e disk model of the hyperbolic space, which has the nice feature of preserving angles. Anyway, one sees the necessity of infinite side-lengths by simply realizing that $a_\text{hyp.tr.}=\pi \Rightarrow \alpha_n=0\,\forall n$, and by straightforwardly applying the hyperbolic law of cosines (\autoref{eq_hypcosinv}). } 
Nonetheless, an $\SO$ element representing a rotation around some fixed oriented axis 
%(recall that the orientation -- i.e. the sign -- of the axis is fixed by the triple-product conditions) 
is generally between 0 and $2\pi$, which means that the areas encoded in the holonomies $O_\ell$ generally range over these values. Spherical triangles achieve this full range of areas, but standard hyperbolic triangles do not. Is there something forcing the areas to be smaller than $\pi$ when the determinant of the Gram matrix, seen as a function of the four holonomies, is negative? It is not hard to find examples showing that there is not. Consequently, we need to make sense of hyperbolic tetrahedra with face ``areas'' in the full range $(0,2\pi)$. Inspired by the observations at the end of the previous section, we look to use triangles stretching across the two sheets. The aim of this section is to describe these new two-sheeted hyperbolic triangles and tetrahedra.

\begin{figure}%
\begin{center}%
\includegraphics[width=.95\textwidth]{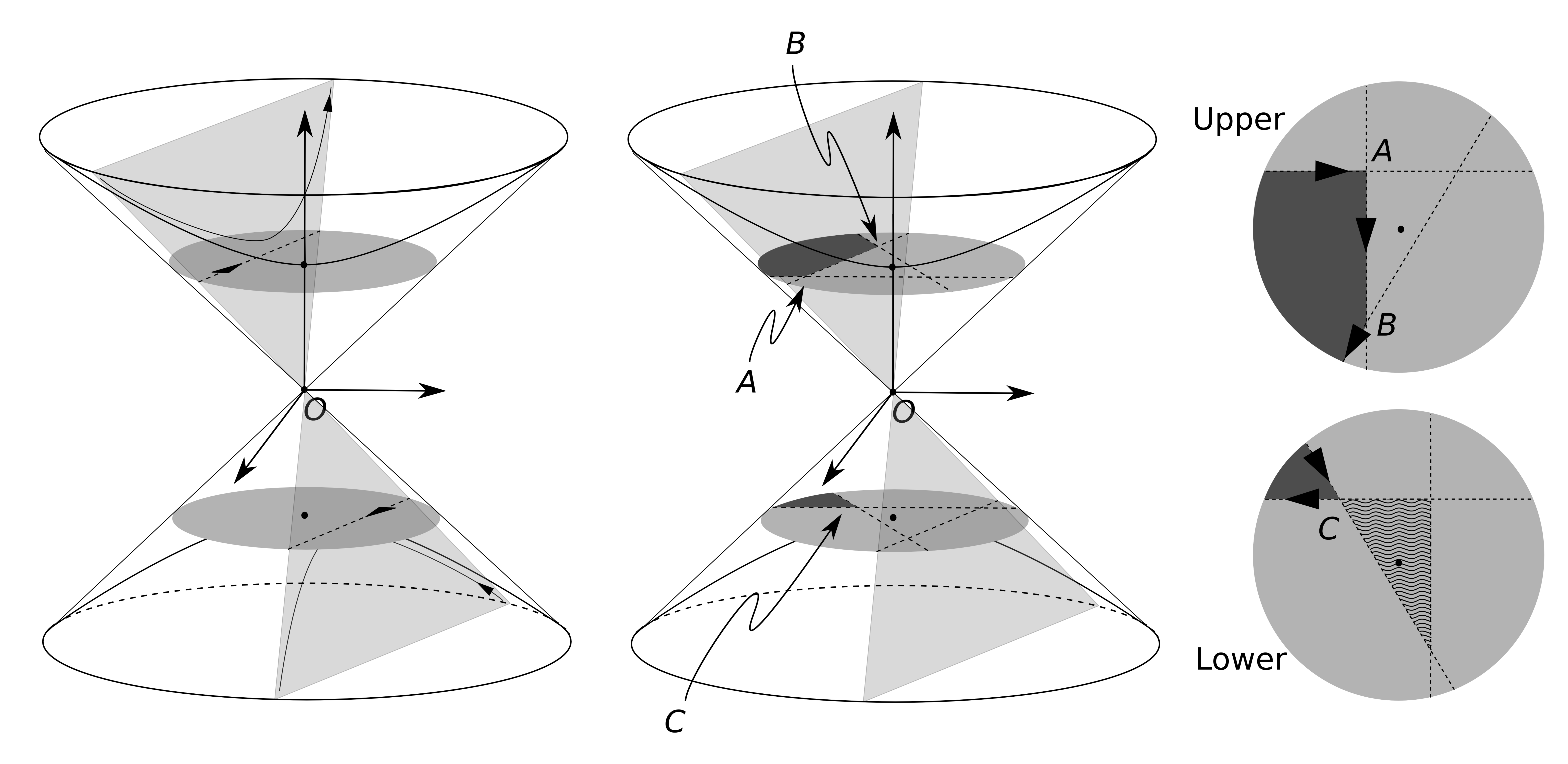}
\caption{The two-sheeted hyperboloid and triangle. \textbf{Left} A plane passing through the origin of the embedding space $\mathbb R^3$ (light gray) intersects the hyperboloid along two geodesics. The dark gray disk gives the Beltrami-Klein model of 2-dimensional hyperbolic space. Geodesics on the hyperboloid are mapped onto straight lines of the Beltrami-Klein disk (dashed lines). In contrast to the Poincar\'e disk model, the Beltrami-Klein model does not preserve angles. \textbf{Center and Right} In dark grey, a two-sheeted triangle. The right-most figure shows the two Beltrami-Klein disks as seen from the origin of $\mathbb R^4$, therefore a positively oriented triangle has a right-handed down-ward pointing orientation with respect to the plane of the page. The central, one-sheeted triangle in the lower sheet is shaded for future reference.}%
\label{fig_twosheeted}
%\end{center}%
%end{figure}%
\vspace{4em}
%\begin{figure}%
%\begin{center}%
\includegraphics[width=.95\textwidth]{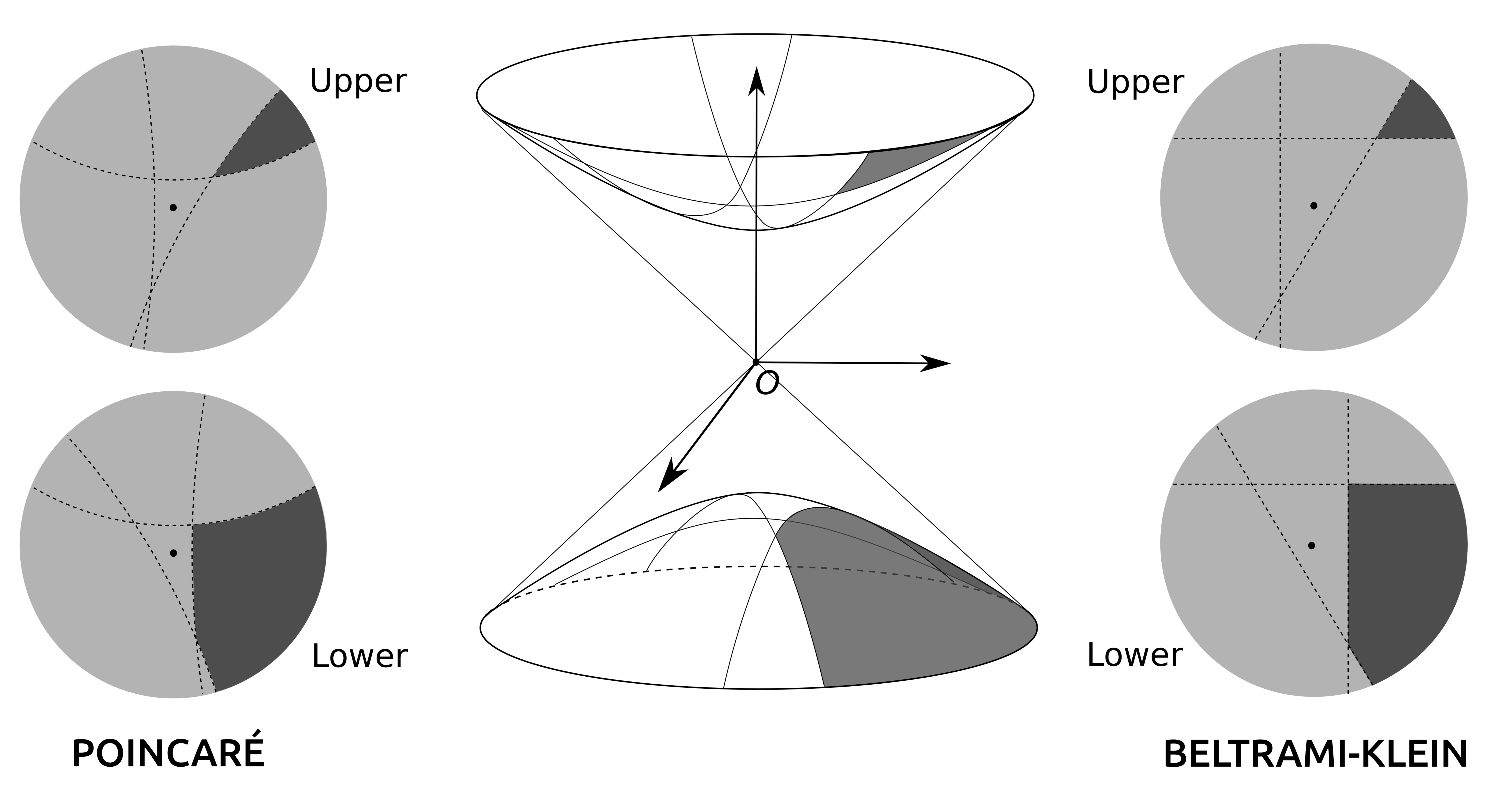}
\caption{Another representation of  a two-sheeted triangle. In the central image the shape of the actual hyperbolic triangle is highlighted. On the sides we represent the triangle in both the Poicar\'e and Beltrami-Klein models. The Poincar\'e model has the virtue of preserving angles, while the Beltrami-Klein models represents geodesics via straight lines and is easily recovered via the intersecting-plane construction shown in the previous picture.}%
\label{fig_twosheeted_bis}
\end{center}%
\end{figure}%

We start with the 2-dimensional triangles. In Figures \ref{fig_twosheeted} and \ref{fig_twosheeted_bis} we have illustrated what we mean by a two-sheeted triangle, and how to orient them. The key idea, is to use the planes passing through the origin of the embedding space $\mathbb R^3$ to extend the geodesics beyond infinity to the other sheet, and to use the natural orientation of the hyperbolae provided by Lorentz boosts which is also consistent with the orientation induced by that of the planes. Figures \ref{fig_twosheeted} and \ref{fig_twosheeted_bis} represent a two-dimensional hyperbolic geometry, and hence the geometry of the faces of a hyperbolic tetrahedron. In three dimensions, the Beltrami-Klein disk model becomes a three-ball model, in which the two-dimensional hyperboloids where the faces of the tetrahedron lie are mapped onto flat disks inscribed in the three-ball; it is one of these that is pictured.

The area of a two-sheeted triangle, however, is not just larger than $\pi$, it is actually infinite. Nonetheless, what appears implicitly in the closure equation is not the area of the triangle, but the total deficit angle perceived by an observer going around it. In an homogeneously curved geometry this happens to be proportional to the area. Therefore, by \textit{defining} a notion of holonomy around a two-sheeted triangle, we effectively provide a notion of ``renormalized'' area for these triangles. At the end of this section, we briefly comment on how far this idea can be pushed.

In order to define a holonomy around a two-sheeted triangle, it is enough to give a prescription for the parallel transport through infinity from one sheet to the other. In other words, one needs to identify the tangent spaces at the point $P$ and $P'$ on the boundaries of the two Poincar\'e or Beltrami-Klein disks, or balls in three dimensions. However, given a geodesics and its extension to the other sheet, there is a very natural prescription for the identification of the aforementioned tangent spaces (see the left columns in each panel of  \autoref{fig_paralleltr_infty}, as well as \autoref{fig_twosheeted}). %
\begin{figure}%
\begin{center}%
\includegraphics[width=.95\textwidth]{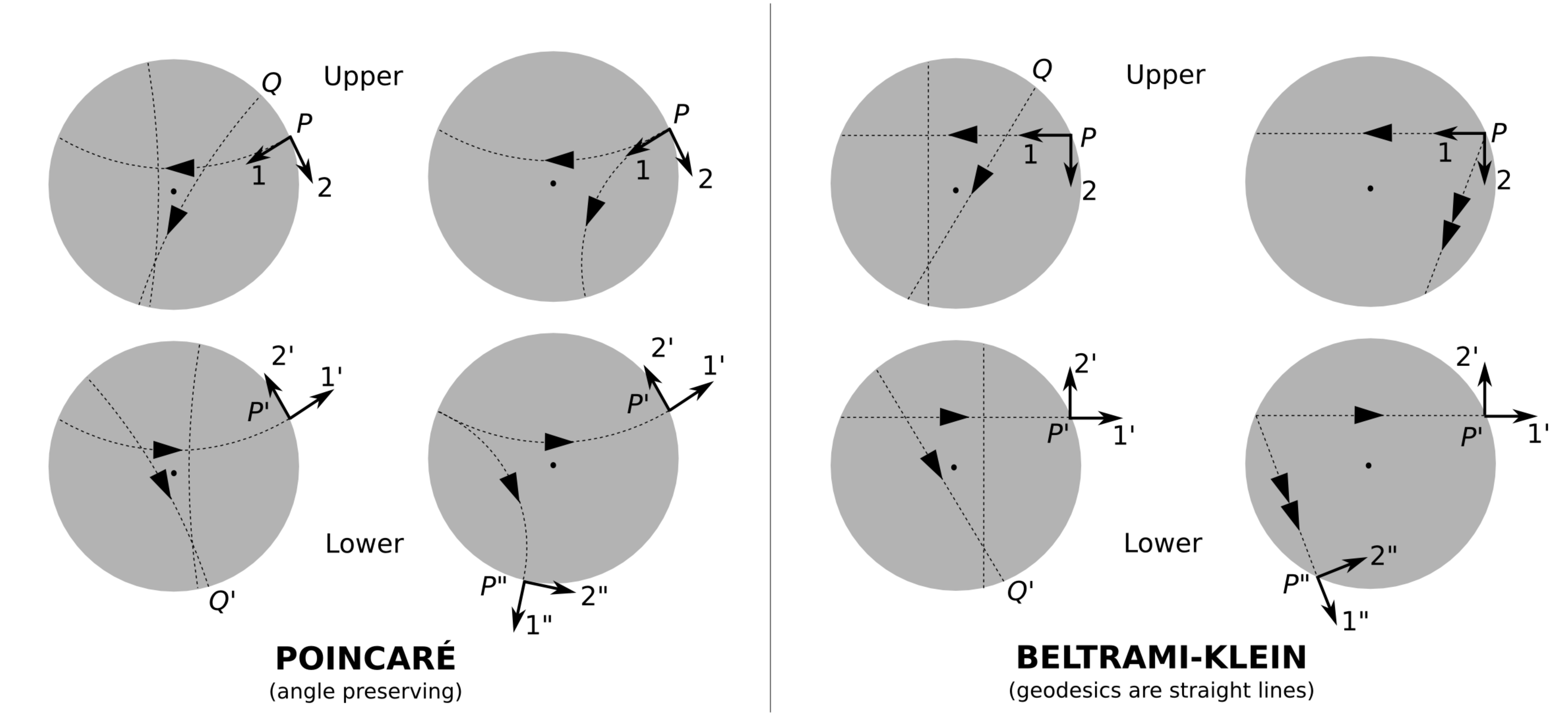}
\caption{The identification across the two sheets of boundary points along the \textit{same} geodesics. The right column of each panel shows what happens when $Q$, belonging to a different geodesics, is let ``collapse'' onto $P$: the point $P$ has to be identified with two different points according to which geodesics they belong to. In this sense, one is not allowed to think of ``gluing'' the two discs together to perform the identification. The arrows, show how a frame is parallel transported across the two disks. Arrow 2 represents the projection onto the disk of a normal to the plane of \autoref{fig_twosheeted}.}%
\label{fig_paralleltr_infty}
\end{center}%
\end{figure}%
%
%This requires that: (\textit{i}) the velocity vector along the geodesics is sent into the velocity vector along its continuation, (\textit{ii}) each vector orthogonal to the velocity stays orthogonal to it and keeps lying on the same two-sheeted 2-dimensional hyperboloid, (\textit{iii}) \textcolor{red}{orientations do not change, neither among the vectors, nor between the vectors and the two-sheeted hyperbolid} {\bf is the third condition clear and precise???}. In particular, a vector pointing toward the interior of a particular two-sheeted triangle, will do so on both sheets.

 This requires that: (\textit{i}) the velocity vector along the a geodesic going out to infinity is identified with the incoming velocity vector on the geodesic's continuation, (\textit{ii}) the vector normal to the outgoing geodesic and pointing towards the interior of a two-sheeted triangle is identified with the only vector with both these properties at the entering point of the incoming geodesics on the other sheet. The second requirement simply preserves the notions of in and outside. In the embedded picture, it requires the vector normal to the geodesic to lie on the same side of the hyperplane defining the geodesic itself. Notice that by orienting the normals to the upper and lower sheets as future and past pointing respectively, the three-dimensional frame composed by the velocity, the normal vector to hyperplane, and the normal to the hyperboloid preserves its orientation thanks to this requirement. This construction can be generalized to the three-dimensional two-sheeted hyperboloid, by considering the normals to the flatly embedded surfaces defining the faces of the two-sheeted tetrahedron instead of the normal to the geodesic arcs defining the sides of the two sheeted triangles.

Does the identification of $P$ and $P'$ and their tangent spaces obtained while moving \textit{along a given geodesic} induce an identification of the boundaries of the two Beltrami-Klein disks (balls)? No. The reason is shown in the right columns of each panel of \autoref{fig_paralleltr_infty}: a point $P$ on the upper sheet is identified with different points on the lower sheet depending on the geodesics through which the point is reached. Therefore, specializing to the relevant 3-dimensional case, the parallel transport prescription we give, instead of identifying the boundaries of the two balls $\partial \mathrm{BK}^3_\text{Upper}$ and $\partial \mathrm{BK}^3_\text{Lower}$, provides a 1-to-1 map between the spaces $\partial \mathrm{BK}_\text{Upper}\times \rmS^2$ and  $\partial \mathrm{BK}_\text{Lower}\times \rmS^2$. Here $\rmS^2$ labels the space of geodesics based at a point on $\partial \mathrm{BK}_\text{Upper,Lower}$.

Having fixed the parallel transport prescription, finding the holonomy around a triangle is just a matter of calculation. A particularly simple way to find the holonomy in the standard case is to note that the parallel transport along a geodesic is trivial, and the only non-trivial contributions come from the ``kinks'' at the vertices of the triangle. Pleasantly, this remains true here because nothing happens when parallel transporting a frame across the two sheets. If the triangle lies completely within one sheet (or on the surface of a sphere) each kink contributes to the final holonomy with a rotation (around the normal to the surface) through an angle $-\tilde\alpha$, where $\tilde\alpha$ is the angle between the velocity vectors before and after the kink. After circuiting a triangle the total rotation amounts to $- \left(3\pi - \sum_{n=1}^3 \alpha_n\right)$, with $\alpha_n=\pi-\tilde\alpha_n$ being the \textit{internal} angles of the triangle. Then, in the case of a spherical triangle we simply obtain its area (modulo $2\pi$) $a= \left(\sum_{n=1}^3 \alpha_n - \pi\right)$. Similarly, for a one-sheeted hyperbolic  triangle, we obtain (again modulo $2\pi$) minus its area $a_{\text{1s}}=\left(\pi - \sum_{n=1}^3 \alpha_n\right)$. However, if the triangle is hyperbolic and two-sheeted, we \textit{define} its ``renormalized'' area through the parallel transport prescription we just outlined, obtaining the formula:
\begin{align}
a_\text{2s}:= 3\pi-\sum_{n=1}^3\alpha_n\;.
\end{align}
To avoid confusion, we will call $a_\text{2s}$ the \textit{holonomy area} of the two-sheeted triangle.  This area is in the range $a_\text{2s}\in(0,2\pi)$. To show this note that the internal angles of the two-sheeted triangle are related to those of the unique (up to congruence) one-sheeted triangle identified by continuing its geodesic sides, see the rightmost, lower panel of Figure \ref{fig_twosheeted}, where the relevant one sheeted triangle is dashed. Calling the angles of the latter triangle $\alpha$, $\beta$, and $\gamma$, where $\gamma$ is the only angle the two triangles have in common, and its area $a_\text{1s}$, one finds:
\begin{align}
a_\text{2s} &= 3\pi - (\pi-\alpha) - (\pi-\beta) - \gamma  = \pi + \alpha + \beta -\gamma  \notag\\
& = 2\pi - a_\text{1s} - 2\gamma \qquad\qquad\qquad\qquad\qquad\qquad < 2\pi \\
& = a_\text{1s} + 2\alpha + 2\beta \qquad\qquad\qquad\qquad\qquad\qquad >0 \notag\,.
\end{align}
The same result could have been obtained by using the simple observation that the holonomy area of a (necessarily two-sheeted) hyperbolic lune of width $\gamma$ is $a_\text{hyp.lune}= 2\pi-2\gamma$ (to be compared to the spherical case: $a_\text{sph.lune}=2\gamma$). We observe that the holonomy areas only make sense modulo $2\pi$, and can only be calculated for regions whose boundaries are arbitrarily well approximated by piecewise geodesics lines. Consequently, it is not possible to make sense of the holonomy area of a full hyperbolic sheet, and the total area of the two sheets is zero, since it is ``enclosed'' by the trivial loop. Nonetheless, given that the starting point of our reconstruction theorem are the holonomies themselves, and not arbitrary regions of the two-sheeted hyperboloid, these definitions are appropriate and useful.

The generalization of this construction to higher dimensions, and in particular to two-sheeted tetrahedra, is straightforward: these tetrahedra are regions of the two-sheeted 3-hyperboloid identified by four points on it, the vertices, and delimited by the intersections of the hyperboloid with the hyperplanes generated by triplets of vertex 4-vectors. Note that to completely characterize the tetrahedron, one has to specify the orientations of the planes. 

%Before stating and proving the general reconstruction theorem in the next section, we observe that the need of introducing the two-sheeted hyperbolic geometries comes not only from the ``area-larger-that-$\pi$'' problem discussed at the beginning of this section. Indeed, it is easy to show that from the Gram matrix produced in the reconstruction it is possible to unambiguously calculate the position of the vertices of the putative tetrahedron associated to that Gram matrix,\footnote{Up to normalization of its columns, this is simply given by the inverse of (the transpose of ) the matrix $N$ introduced in the previous section. See next section.} and such vertices are generally not all contained in one single sheet. (See e.g. \cite{FengLUO1997}).

%-----------------------------------------------------------------------------------------
\section{Curved Minkowski Theorem for tetrahedra}\label{sec_theorem}

Now that the geometric picture has been clarified, we can state and finally prove the curved Minkowski theorem for tetrahedra. When we want to emphasize that the Gram matrix can be caclulated directly from the holonomies $\{ O_\ell \}$, e.g. using Eq. \eqref{eq_scalarOO} and related expressions, we write $\Gr(O_{\ell})$. 

\begin{theorem}\label{thm_mink}
Four $\SO$ group elements $O_\ell$, $\ell=1,\dots,4$ satisfying the closure equation $O_4 O_3 O_2 O_1 = \mathrm{e}$, can be used to reconstruct a unique generalized (i.e. possibly two-sheeted in the hyperbolic case) constantly-curved convex tetrahedron, provided:
\begin{itemize}
\item[\textnormal{(i)}] the $\{O_\ell\}$ are interpreted as the Levi-Civita holonomies around the faces of the tetrahedron,
\item[\textnormal{(ii)}] the path followed around the faces is of the so-called ``simple'' type (see \autoref{sec_strategy}), and has been uniquely fixed by the choice of one of the two couples of faces (24) or (13),
\item[\textnormal{(iii)}] the orientation of the tetrahedron is fixed and agrees with that of the paths used to calculate the holonomies,
\item[\textnormal{(iv)}] the non-degeneracy condition $\det\Gr(O_\ell) \neq 0$ is satisfied.
\end{itemize}
The uniqueness is understood to be modulo isometries. 

%\textnormal{\bf Note:} The theorem can be equivalently formulated in terms of four closing $\mathrm{PSU}(2)$ group elements. In this case, the group elements have to be interpreted as the spin-connection holonomies around the faces of the tetrahedron up to their sign, which,  in turn, can be uniquely determined as discussed above. \textcolor{red}{do we want to keep this comment? is it ok? we loose sign on the rhs though.}
\end{theorem}

In particular, condition \textit{(i)}  means that the $O_\ell$'s  written in the form $\exp \left\{ \pm a_\ell \hat n_\ell \cdot \vec J\;\right\}$  have the following geometrical interpretation: (1) the $a_\ell$ are the areas
%\footnote{Possibly understood as a ``holonomy areas'' as explained in the previous section.} 
of the faces of the tetrahedron (possibly interpreted as holonomy areas), and (2) the $\hat n_\ell$ are the outward pointing normals to these faces when parallel transported (along the simple path chosen) to a common reference frame. Also, it turns out that: (3) the tetrahedron has positive (negative) curvature if  $\det\Gr(O_\ell) >0$ ($\det\Gr(O_\ell)<0$, respectively); (4) the tetrahedron is double-sheeted if it has a negative curvature and the cofactors of the Gram matrix do not agree in sign. The proof is an extension of the formalism appearing at Eq. \eqref{eq_Ws} below.

Observe that the four conditions to be satisfied for the theorem to hold have distinct characters: condition \textit{(i)} is key to the theorem, it allows its geometric interpretation; condition \textit{(iii)} is simply needed to avoid the possibility of reconstructing the parity reversed tetrahedron as well; condition \textit{(iv)} is technical and, unfortunately, can be cumbersome from the point of view of the holonomies, since the Gram matrix is a nice object geometrically speaking, but not as simple algebraically; finally, condition \textit{(ii)} has a somewhat strange status. Indeed, a condition of this type is certainly needed to take care of the parallel transport ambiguities present in the curved setting, but at the same time the specific form we are employing looks quite arbitrary---even if inspired by a simplicity criterion---and in principle can be modified to other choices of paths, which would, in turn, require a few somewhat obvious modifications in the reconstruction procedure. The simple-path condition naturally arises in the four-dimensional context of \cite{HHKR}.

Before giving the proof of the main theorem we give a short proof of a useful lemma: 
\begin{lemma}\label{lemma_minors}
The principal minors of $\Gr$ are positive, with the exception of the $4\times4$ minor in the hyperbolic case.
\end{lemma}
\begin{proof}
The $1\times1$ principal minors are immediate, since each is equal to 1. The $2\times2$ minors are also easily seen to be positive since they are equal to  $(1-\cos^2\theta_{\ell m})$, for the appropriate choice of indices $(\ell,m)$. Finally, to show that also the $3\times3$ minors are all positive, consider first the case of the principal minor $m_{4}$ equal to the determinant of the matrix obtained by erasing row and column 4 from $\Gr$:
\begin{align}
m_4 = \det \Big( \hat n_\ell . \hat n_m \Big) = \left[\det \left( \hat n_1 \Big| \hat n_2 \Big| \hat n_3 \right) \right]^2
\end{align}
where the unit vectors $\hat n_\ell$ are those appearing in Eqs. \eqref{eq_O} and \eqref{eq_Ohyp} (with signs fixed by the triple product criterion),\footnote{In the hyperbolic case, the triple product criterion described at the beginning of \autoref{thm_mink} gives the $\hat n_\ell$ signs opposite to the geometric ones. However, the Gram matrix is unaffected by this global change in sign. } and the matrix appearing at the furthest right is the matrix which has the three 3-vectors $\hat n_\ell$ as columns. In light of this formula $m_4$ is trivially positive. The same holds for $m_2$. A little more effort is needed to prove that $m_1$ and $m_3$ are also positive. Explicitly:
\begin{align}
m_3 = \det\left(
\begin{array}{ccc}
1 & \hat n_1 . \hat n_2  & \hat n_1 . \hat n_4 \\
 & 1 & \hat n_2 . \bO_1 \hat n_4 \\
 \text{SYM} &  & 1 
\end{array}
\right) = \left[\det \left( \hat n_1 \Big| \hat n_2 \Big| \bO_1\hat n_4 \right) \right]^2,
\end{align}
where in the first equality we used the definition of Eq. \eqref{eq_dihedral}, which takes into account the parallel transport of $\hat n_4$ to vertex 4 along the special edge; while in the second we made use of the fact that $\bO_1 \hat n_1 = \hat n_1$, and therefore $\hat n_1 . \hat n_4 = \hat n_1 . \bO_1\hat n_4$. Therefore, $m_3$ is positive. It can be shown that $m_1$ is positive by a very similar argument. 
\end{proof}

The proof of the theorem proceeds in a completely constructive way, and without loss of generality, it is performed within the explicit choice of edge $(24)$ being the special one. Most of the steps necessary for the reconstruction were explained in great detail in \autoref{sec_strategy}, and will not be discussed again. Our attention is focused on the well-definedness and unambiguous statement of each step of the reconstruction. We will also prove the consistency of the reconstruction procedure. Therefore the theorem is subdivided into two parts: in the first, we show that the $O_\ell$ uniquely identify a Gram matrix that, in turn, is associated to a unique curved tetrahedron; in the second part we show that the Levi-Civita holonomies around the four faces of the tetrahedron are necessarily given by the $O_\ell$ themselves. Loosely speaking, in the first part we extract from the closure relation and the simple-path condition the dihedral angles of a tetrahedron which uniquely determine it, and in the second we verify that the areas of the reconstructed tetrahedron are necessarily the same as those encoded in the initial group elements $O_\ell$.

\begin{proof} $\phantom{ciao}$\\
\textbf{Part one}
First, calculate the triple products appearing in Eq. \eqref{eq_triple} using the group elements via Eq. \eqref{eq_tripleOOO} (properly generalized in the way discussed in the first section for $\{\ell,m,q\}=\{ 1,2,4\}$ or $\{2,3,4\}$), and fix the signs $\pm_\ell$ appearing there by requiring  these four triple products to be positive (note that there is only one such choice). Geometrically, this completely fixes the signs of the normals by imposing the convexity of the tetrahedron.\footnote{Note that the so reconstructed normals would turn out to have the opposite sign with respect to the geometric ones in the hyperbolic case. This global flip in the sign of the normals does not compromise any of the following steps.} This allows the unambiguous specification of the entries of the (putative) Gram matrix 
\begin{align}
\Gr_{\ell m} :=\cos \theta_{\ell m}\quad\text{for $\ell\neq m$,} \quad \text{and} \quad \Gr_{\ell\ell} =1 \,,
\end{align}
with the right-hand side of the first equation being calculated via Eq. \eqref{eq_scalarOO} and its generalization for $\{\ell,m\}=\{2,4\}$. We stress that $\Gr$ is a function of the $O_\ell$'s only. Now,
\begin{align}
\text{either }\quad \sgn \det\Gr >0 \,,\quad \text{ or } \quad \sgn\det\Gr<0\,,
\end{align}
since the null case has been excluded by hypothesis. Define the $4\times4$ matrix $g=\text{diag}(\sgn\det\Gr,1,1,1)$, to be interpreted as the metric of the four-dimensional embedding space as described in \autoref{sec_gram}. Then, there exist four 4-vectors $N_\ell$ such that
\begin{align}
\Gr_{\ell m} = \sum_{\mu,\nu} g_{\mu\nu} N^\mu_\ell N^\nu_m ,
\end{align}
or more symbolically $\Gr= N^TgN$. In particular, $g_{\mu\nu}N^\mu_\ell N^\nu_\ell = 1$, and the four 4-vectors $N_\ell$ can be interpreted geometrically as the oriented unit normals to the hyperplanes passing through the origin of $\mathbb R^4$ which, upon intersection with the the unit sphere $\rmS^3_u \subset \mathbb R^4$ (unit two-sheeted hyperboloid $\rmH^3_u$, respectively), identify the great spheres (great hyperboloid, respectively) bounding the tetrahedron itself. {See the figures and discussion of  \autoref{sec_gram}.}

The vertices of the tetrahedron are located along the intersections of the triplets of hyperplanes normal to the $N_\ell$. Hence the matrix $W:=-(N^{-1})^T$ has columns $W_\ell$ proportional to the 4-vectors identifying the vertices of the tetrahedron (the minus sign in this formula fixes the correct sign of the vertex vectors). We define $V_\ell := W_\ell / \sqrt{|(W_\ell)^T g W_\ell |}$. In the spherical case, the vertex vectors $V_\ell$ completely characterize the tetrahedron;  they identify four points on the unit sphere $\rmS^3_u$ that can be connected by the shortest geodesic segments between them. However, in the hyperbolic case it is not \textit{a priori} clear that the $V_\ell$ intersect the two-sheeted unit hyperboloid $\rmH^3_u$. Indeed, for them to do so, they must be timelike, that is they must satisfy $V_\ell ^T \eta V_\ell = -1$. However, this is equivalent to the condition $W_\ell^T \eta W_\ell <0$, which in turn must be true because of the following relations and the result of  \textcolor{blue}{Lemma} \autoref{lemma_minors} (which states $m_\ell>0$ for all $\ell$):
\begin{align}
\label{eq_Ws}
\left(W_\ell\right)^T \eta W_\ell = \left( W^T\eta W \right)_{\ell\ell} = \left(\left( N^T \eta N \right)^{-1}\right)_{\ell\ell}  = \left(\Gr^{-1}\right)_{\ell\ell} = (\det\Gr)^{-1} m_\ell <0,
\end{align}
where, recall, $m_\ell$ is the principal minor obtained by erasing row and column 4 from $\Gr$. To obtain the last equality, the fact is used that being a diagonal minor, $m_\ell$ is also equal to the $(\ell,\ell)$-th cofactor of $\Gr$. Therefore, we can conclude that also in the hyperbolic case a unique generalized (i.e. possibly two-sheeted) tetrahedron can be identified. It suffices to define the ``shortest'' geodesic between two vertices as the generalized geodesic (i.e. possibly going through infinity) that does not pass through any point defined by the intersection of the hyperboloid and three of the four hyperplanes normals to the $\{N_\ell\}$ other than its initial and final points.
 This concludes the first part of the proof.\\

\textbf{Part two} 
The group elements $O_\ell$ and closure relation Eq. \eqref{eq_closure} specify more data than the Gram matrix alone. Thus, we have to verify the consistency of all of this data. Indeed, the construction from part one guarantees only that the dihedral angles of the reconstructed tetrahedron are compatible with the holonomy group elements $O_\ell$, but not that the reconstructed areas also match those encoded in the $O_\ell$. More specifically, we have claimed that the $O_\ell$ can be interpreted as holonomies of the Levi-Civita connection around the various faces of the tetrahedron, and this implies (see \autoref{sec_strategy} and \autoref{sec_twosheeted}) that the rotation angles of the $O_\ell$ are the areas of the faces of the tetrahedron. We now prove this claim.

Given the reconstructed tetrahedron, one can explicitly calculate the holonomies along the specific simple path on its 1-skeleton used in the reconstruction. Call these the reconstructed holonomies, $\tilde O_\ell$. Although they satisfy $\tilde O_4\tilde O_3\tilde O_2\tilde O_1=\mathrm{e}$, and their 4-normals satisfy $\tilde{N}_\ell=N_\ell$ by construction, it is not yet clear whether the $\tilde O_\ell$ are \textit{necessarily} equal to the $O_\ell$ (up to global conjugation, i.e. gauge). Demonstrating this is what we mean by showing consistency of the reconstruction. 
%It can also be rephrased in the following way: four group elements $\tilde O_\ell$, satisfying $\tilde O_4\tilde O_3\tilde O_2\tilde O_1=\mathrm{e}$, and whose related Gram matrix is given in the sense specified in the reconstruction, must be $\tilde O_\ell = O_\ell$ up to a global conjugation by an element of $\SO$. 
We once more proceed constructively, and show that both $\tilde{\hat n}_\ell = \hat n_\ell$ and $\tilde a _\ell = a_\ell$, in the notation of Eqs. \eqref{eq_O} and \eqref{eq_Ohyp}. We will show that the Gram matrix and the closure equation contain all the information needed to completely fix the $\tilde O_\ell$. Because the $\tilde O_\ell$ and the $O_\ell$ have the same Gram matrix we will briefly drop the distinction and omit the tildes.

 First align $\tilde{\hat n}_3$ with $\hat n_3$ by acting with a global rotation (conjugation). A second global rotation around the $\hat n_3$-axis can be used to align $\tilde{\hat n}_1$ with $\hat n_1$; this is always possible because $\tilde{\hat n}_1 . \tilde{\hat n}_3=\tilde\Gr_{13}=\Gr_{13}={\hat n}_1 . {\hat n}_3$. Now the system is completely gauge-fixed and there is no further freedom to rotate the vectors. 
 %Then, we simply need to show that all the other vectors agree, as well as the areas. 
  The vector $\hat n_2$ has a fixed angle with both $\hat n_1$ and $\hat n_3$, determined by $\Gr_{12}$ and $\Gr_{23}$, and there are \textit{a priori} at most two vectors with this property (identified by the intersection of two cones around $\hat n_1$ and $\hat n_3$, respectively). However, only one of those satisfies the additional requirement that $(\hat n_1 \times \hat n_2 ) .\hat n_3 >0$, which was crucially used in the reconstruction.\footnote{Notice, that existence in not in question, since it is guaranteed by construction. Only uniqueness needs an argument.} Similarly,  $\hat n_4$ is also uniquely determined. All that remains then is to show that the entries of the Gram matrix completely fix the areas.
  
  Consider $\Gr_{24}=\hat n_2 . \bO_1 \hat n_4$. Since $\hat n_1$, $\hat n_2$, and $\hat n_4$ are all given, there exists at most two values of $a_1$ (in the interval $(0,2\pi)$) that solve this equation (geometrically this is again the intersection of two cones). The triple product condition $(\hat n_2 \times \hat n_1 ) .\bO_1\hat n_4 >0$ singles out one of these two solutions.  Similarly, one fixes $a_3$ by using the analogous expression $\Gr_{24}=\hat n_2 . \bO_3^{-1} \hat n_4$ and $(\hat n_3 \times \hat n_2 ) .\bO_3^{-1}\hat n_4 >0$. To conclude, we need to show that $a_2$ and $a_4$ are completely determined.
  
%   To do this, one could adventure into some algebraic gymnastic using explicitly the closure equation. However, there is a more clever way to proceed. 
Consider the closure equation $O_4' O_3' O_2' O_1' \equiv( O_3^{-1}  O_4 O_3 ) O_2 O_1 O_3= \mathrm{e}$, where we identify $O_4' \equiv O_3^{-1}  O_4 O_3 $, $O_3'\equiv O_2$, and so on. We have completely fixed $O_1$ and $O_3$, as well as $n_2$ and $n_4$. The remaining unknowns are $a_2$ and $a_4$. In the language of the new closure one needs only determine $a'_3$ and $a'_4$. The Gram matrix of the new closure is the same as the previous one if edge $(24)'$ is selected as the new special edge, and is therefore completely known. Following the same construction then we can fix $a'_1$ and $a'_3$, but these are respectively the same as $a_3$ and $a_2$. Therefore we have fixed $a_2$.\footnote{Again we do not discuss existence of these solutions, only their uniqueness, since existence was covered in the proof's first part.} Now that only one variable is left, an explicit use of the closure equation clearly fixes it uniquely, by giving explicit expressions for both $\cos a_4$ and $\sin a_4$.
\end{proof}

Note that in the second part of the theorem, the spherical and the hyperbolic cases (even the two-sheeted one) are treated uniformly. In fact, once the details of the reconstructed tetrahedron are given, one only needs a parallel transport rule (and a path) to write down a closure equation and associate it to a Gram matrix consistent with the reconstruction. This works straightforwardly in each of the cases.

\part{Phase space of shapes}

\section{Curved tetrahedra, spherical polygons, and flat connections on a punctured sphere}\label{sec_relations}
In the first part of this paper, we have shown how four $\SU$ holonomies satisfying a closure constraint (and a non-degeneracy condition) give rise to the geometry of a curved tetrahedron
%\footnote{Notice we are not requiring the correspondence to be one-to-one.} 
embedded in either $\mathrm S^3$ or $\mathrm H^3$. This closure admits at least two other interpretations: the non-trivial holonomies of a flat connection on a quadruply punctured 2-sphere satisfy such a closure; and this constraint can also be associated to the four sides of a geodesic polygon embedded in $\mathrm S^3\cong\SU$. The flat-connection viewpoint is important, has attracted much attention in the literature, and is closely connected to the motivations for our work (see \cite{HHKR}).

The moduli space of flat connections on a punctured Riemann surface has a natural phase-space structure \cite{Atiyah1983,Goldman1984,Jeffrey1994} that can be deduced, for example, via gauge-theoretic arguments. In this framework the final, finite-dimensional phase space is obtained after a reduction by the infinite-dimensional gauge symmetries of the initial theory. A completely finite-dimensional approach to the problem was put forward by Anton Alekseev, Yvette Kosmann-Schwarzbach, Anton Malkin, and Eckhard Meinrenken \cite{Alekseev1998,Alekseev2000,Alekseev2002}, who built generalized phase-space structures associated to each puncture and handle of the Riemann surface. These spaces are then ``fused'' together in order to obtain the usual phase-space structure, after a further reduction by a global topological constraint. These generalized structures are well adapted to the polygonal interpretation of the closure constraint, and allow the association of a natural phase-space to the polygons in $\mathrm S^3$ of fixed side lengths. This was the content of the work of Thomas Treloar \cite{Treloar2000}, who generalized the previous constructions of phase spaces of polygons on $\mathrm E^3$ \cite{Kapovich1996} and $\mathrm H^3$ \cite{Kapovich2000} to the compact space $\mathrm S^3$. The novelty of the work of Alekseev and collaborators, which is reflected in the spherical-polygon case, is the fact that one is forced to abandon Poisson structures and to step into the realm of \textit{quasi}-Poisson structures, for which the Jacobi identity is violated by a specific term.

The violation of the Jacobi identity is quite a drastic change, but it cannot be avoided if you are to introduce genuinely group-valued moment maps \cite{Alekseev1998,Alekseev2000,Alekseev2002}. Indeed, in Alekseev and collaborators' framework the topological (closure) constraint is equivalent to fixing the total, group-valued momentum of the system to the identity;
%\footnote{Clearly, in a group the ``vanishing'' element in the above sense is the identity element. } 
this is closely analogous to the standard procedure of setting the relevant algebra-valued momentum to vanish when it generates gauge transformations. In this language, the generalized closure constraint is better understood as a deformation of the Gau{\ss} constraint of gauge theories, see also the discussion of spin-networks in sections \ref{sec_Intro} and \ref{sec_outlook}. Interestingly, the violation of the Jacobi identity becomes irrelevant after the reduction to the gauge invariant space is performed.

We believe these fundamental ideas about symmetry may provide an important qualitative shift in thinking about the cosmological constant in physics \cite{HHKR}. So, in this part of the paper we present this material as constructively and intuitively as we can and whenever possible connect the mathematical formalism to the physicists' language. Our focus will be on the tetrahedral interpretation of the closure constraint, which is a novel feature of our work, and hence many considerations specific to this interpretation will be put forward. In particular, our interpretation of the phase space we construct is in terms of a phase space of shapes for curved tetrahedra.

A peculiar feature of our construction, seemingly coincidental, is that for $\SU$ it happens that the Jacobi identity is actually satisfied also at the level of a single puncture's generalized phase space.

%-----------------------------------------------------------------------------------------
\section{Quasi-Poisson structure on $\SU$}\label{sec_qpoisson}

Before considering the phase space of curved tetrahedra, we start with the simpler problem of defining a quasi-Poisson structure for each face. This is analogous to the construction of the phase-space structure on the moduli space of flat connections on a punctured sphere out of the quasi-Poisson structures associated to each puncture.

As mentioned in section \ref{sec_Intro}, an important feature of Minkowski's construction in the flat case is that the closure constraint is also the generator of gauge transformations at each node of the spin network, i.e. it is the generator of rotations in the tetrahedral picture. In particular, each flux generates rotations of the associated face vector. We want to reproduce this feature with the curved tetrahedra. To do so  we need to formalize the flat case. 

%---
\paragraph*{Review of the flat case:} The group $\SU$ acts on a three-vector $\vec a\in\mathbb R^3$ via its vectorial (spin 1) representation. This action can be cast as a Hamiltonian action generated by the three-vectors:
\be
\{a^i, f(\vec a)\} = \frac{\D}{\D t} f\left(\E^{-t \bJ^i}\vec a\right) \big|_{t=0} 
\ee
for any function $f:\mathbb R^3 \rightarrow \mathbb R$. Because $(\bJ^i)^l_{\phantom{j}k} = -\epsilon^{il}_{\phantom{il}k}$, one immediately finds
\be
\{a^i, f(\vec a)\} = \epsilon^{il}_{\phantom{il}k} a^k\frac{\partial}{\partial a^l} f\left(\vec a\right) 
\quad\text{or}\quad
\{a^i, \cdot\} = \epsilon^{il}_{\phantom{il}k} a^k\frac{\partial}{\partial a^l}.
\ee
Applying this to the function $f(\vec a) = a^j$ yields 
\be
\{a^i,a^j\} = \epsilon^{ij}_{\phantom{ij}k} a^k.
\label{eq_structuresu2}
\ee
However, it is useful to explore this result from a slightly different perspective. Identify $\mathbb R^3$ with the dual $\su^*$ of the Lie algebra $\su$, via $\vec a \mapsto \alpha := \vec a.\vec \eta$, where $\eta^i\in\su^*$ is dual to the basis $\tau_i\in\su$:
\be
\langle \eta^i,\tau_j\rangle = \delta^i_j \;,\;\text{where}\; [\tau_i,\tau_j] = \epsilon_{ij}^{\phantom{ij}k}\tau_k.
\ee
The action of $\SU$ on $\mathbb R^3$ is mapped into the coadjoint action of $\SU$ on $\su^*$:
\be
\alpha_G := (G\triangleright\vec a).\vec \eta = \vec a. (\Ad^*_{G^{-1}} \vec \eta) =\Ad^*_{G^{-1}}\alpha .
\ee
The vector field $y_{\su^*}$ associated to an infinitesimal transformation is
\be
y_{\su^*}=\langle -\ad^*_y\alpha,\partial_\alpha\rangle = \langle \alpha,\ad_y\partial_\alpha\rangle= \langle \alpha, [y,\tau_l]\rangle\frac{\partial}{\partial a_l} = \epsilon_{il}^{\phantom{il}k}y^i a_k\frac{\partial}{\partial a_l}
\label{eq_coadjvectorfield}
\ee
where $y\in\su$ is the infinitesimal version of $G$, and $\partial_\alpha:=\tau_l\frac{\partial}{\partial a_l}$ is an $\su$-valued vector field on $\su^*$. Hence the Poisson brackets on $\mathbb R^3$ that we wrote above can be now interpreted as Poisson brackets on $\su^*$:
\be
\{\langle \alpha, y\rangle, \cdot\} = y_{\su^*}.
\label{eq_Poissongenerator}
\ee
The meaning of this equation is that the function $\langle \alpha, y\rangle$ on $\su^*$ is the Hamiltonian generator of the coadjoint action in the direction of $y\in\su$ on the space $\su^*$.

Notice that in the latter approach the fact is put to the forefront that the dual $\frak g^*$ of a Lie algebra $\frak g$ carries a canonical Poisson structure induced by the Lie brackets on $\frak g$ itself. This is a classical result due to Alexandr A. Kirillov and Bertram Kostant \cite{Konstant1970,Kirillov1976}.

We introduce some useful nomenclature and notation. Define the Poisson bivector
\be
P = P_{ij}\left(\frac{\partial}{\partial a_i}\otimes\frac{\partial}{\partial a_j}-\frac{\partial}{\partial a_j}\otimes\frac{\partial}{\partial a_i}\right)
\ee
so that
%\footnote{Notice that whenever $\{\cdot,\cdot\}$ are Poisson brackets deriving from a symplectic 2-form $\omega$, then $P$ is its inverse.}
\be
P( \D f, \D g) := \iota(P) (\D f \otimes \D g) := \{f , g\} \quad \forall f,g\in \mathcal C^1(\su^*,\mathbb R),
\ee
where $\iota$ denotes contraction. 
The bivector $P$ can also be interpreted as a map from one-forms to vector fields; for this it is enough to contract it with a single 1-form. When viewing it as this  map we denote it $P^\#$:
\be
P^\# : \Omega^1(\su^*) \rightarrow \mathfrak X(\su^*)\;, \;\D f \mapsto P^\#(\D f) \;\text{such that}\; \iota\left(P^\#(\D f)\right)\D g = P(\D f, \D g),
\ee
where $\Omega^n(M)$ is the space of $n$-forms on a manifold $M$ and $\mathfrak{X}(M)$ is the space of vector fields on $M$. 

Now we can rewrite Eq. \eqref{eq_Poissongenerator} as\footnote{Indicating the inverse of $P^\#$ (possibly after its restriction to an appropriate subspace) by $\omega^\flat$:
\be
\D\langle \alpha, y\rangle = \omega^\flat(\langle -\ad^*_y \alpha, \partial_\alpha \rangle) = \iota({\langle -\ad^*_y \alpha, \partial_\alpha \rangle}) \omega,
\ee
where $\omega\in\Omega^2$. This is a well-known formula in the context of symplectic geometry. In a slightly more general framework it goes under the name of the \textit{moment map condition}.  }
\be
P^\#(\D\langle \alpha, y\rangle ) = y_{\su^*}.
\label{eq_Poissongenerator2}
\ee
This formula tests the vector field generated by a \textit{linear} function $\langle \alpha, y \rangle$ of the \textit{Hamiltonian} generators of the group action $\alpha$. The general case is\footnote{It is actually immediate to show that this condition is equivalent to the previous one by the linearity of $P^\#$. The right hand side of this equation can be written in a coordinate free way as $\langle \alpha, [\partial_\alpha f, \partial_\alpha]\rangle$, where again $\partial_\alpha=\tau_i\frac{\partial}{\partial a_i}$ is an $\su$ valued derivative on $\su^*$.}
\be
P^\#(\D f(\alpha) ) = \frac{\partial f}{\partial a_k} (\tau_k)_{\su^*}\qquad \forall f\in\mathcal C^1(\su^*,\mathbb R).
\label{eq_Poissonsu2}
\ee
and will be useful in generalizing to non-linear spaces of Hamiltonian generators. 

If we are given a transformation to implement on $\su^*$, the right-hand side of this equation is fixed via Eq. \eqref{eq_coadjvectorfield}, while postulating its Hamiltonian generators (the $\alpha$ themselves) fixes the argument of $P^\#$. These two pieces of information, taken together, fix uniquely the Poisson bivector.

%-----
\paragraph*{The curved case}  We now adapt this constructive procedure to the curved case. That is, we will deduce the appropriate bracket on the space of generalized $\Su$ area vectors by postulating both the way they transform and the generators of this transformation. In analogy to the flat case, the transformation will act by conjugation and be generated by the $\Su$ area vectors. Important modifications to the flat construction are needed to fully implement this strategy. This will lead us into the subtle realm of \textit{quasi}-Poisson manifolds.

 In the previous sub-section, it was natural to treat the area vectors as elements of $\su^*$. Two steps are needed in order to promote them to elements of $\SU$: identify $\su^*$ with $\su$ in a natural way, and then ``exponentiate'' the result in some manner.

We use the Killing form on $\su$, $K(\cdot,\cdot)$ to implement the first step. Indeed, for any $\alpha\in\su^*$ there exists a unique $x_\alpha \in \su$ such that
\be
\langle \alpha, y \rangle = K(x_\alpha, y) \quad \forall y\in\su.
\ee
Normalize $K$ so that $K(\tau_i,\tau_j)=\delta_{ij}$, then Eq. \eqref{eq_Poissonsu2} is essentially unaltered
\be
P^\#(\D f(x) ) = \frac{\partial f}{\partial x_k} (\tau_k)_{\su}\qquad \forall f\in\mathcal C^1(\su,\mathbb R),
\ee 
except that the coadjoint action is mapped into the adjoint action of $\su$ on itself:
\be
(\tau_k)_{\su} = \langle \eta^i,-\ad_{\tau_k} x \rangle\frac{\partial}{\partial x^i}.
\ee

In order to ``exponentiate'' this result, we need to find a vector field on the group manifod generating the $\SU$-transformations of the face holonomies, i.e. the analogue of $(\tau_k)_{\su} $, and generalize the simple partial derivative of the function $f$ to an appropriate vector field on the non-linear $\SU$ group manifold. The first task is simple, since conjugation of the $\SU$ face holonomies by elements of $\SU$ generalizes the adjoint action of the group on its Lie algebra: 
\be
\Ad_{H^{-1}} x \quad \leadsto \quad \mathrm{AD}_{H^{-1}} G := H^{-1} G H.
\ee
The vector field implementing an infinitesimal transformation is
\be
y_{\su^*}=\langle \eta^i,-\ad_y x \rangle\frac{\partial}{\partial x^i} \quad \leadsto \quad y_{\SU}=y^L - y^R,
\label{eq_ySU2}
\ee
where $y^{R,L}$ are respectively the right- and left-invariant vector fields on $\SU$, with the value $y\in\su\cong \mathrm T_\E\SU$ at the identity. 

More interesting is generalizing the derivative of the function $f$ in the direction associated to a basis element $\tau_k$ of the Lie algebra. There is no unique, natural derivative (vector field) on the group $\SU$ associated with the direction $\tau_k$. This is because the group is non-Abelian and hence non-linear. In particular, derivatives in any direction $y$ can be associated to either left or right translations on the group, translating along $y^R$ and $y^L$, respectively. So, what is the appropriate combination $\hat y$ of these two derivatives? Both $y^R$ and $y^L$ reduce to the usual derivation in the flat (Abelian) limit. Interestingly, the antisymmetry of the Poisson bivector $P$ fixes this ambiguity, selecting
%\footnote{The factor of one-half simply assures that in the flat limit we recover the previous result.} 
$\hat y = \frac{1}{2}( y^L+y^R)$. Indeed, suppose $\hat y = A y^L + B y^R$, with $A+B=1$ to assure the correct flat limit. Then for all functions $f$:
\begin{align}
0\equiv P(\D f \otimes \D f) &= P^\#(\D f) (\D f) = \sum_k\big(\widehat{\tau_k} f\big)\big((\tau_k)_{\SU} f\big) \notag\\
&= \sum_k \left[A (\tau_k)^L \otimes (\tau_k)^L  -  B (\tau_k)^R \otimes (\tau_k)^R  \right](\D f\otimes \D f) + (B-A) \left[(\tau_k)^L f \cdot (\tau_k)^R f \right]\notag\\
& = (A-B) \left\{\sum_k (\tau_k)^L f \cdot \left[(\tau_k)^L f -  (\tau_k)^R f \right]\right\} 
\quad \Longrightarrow \quad A=B=\frac{1}{2},
\end{align}
where we used the identity $\sum_k (\tau_k)^L\otimes(\tau_k)^L = \sum_k (\tau_k)^R\otimes(\tau_k)^R$. \\

Thus, we have obtained the following condition on the quasi-Poisson bivector $P$ on $\SU$:
\be
P^\#(\D f) = \frac{1}{2}\left[ \left( (\tau_k)^L + (\tau_k)^R \right)f \right] (\tau_k)_{\SU} \qquad \forall f\in\mathcal C^1(\SU, \mathbb R).
\label{eq_PoissonSU2}
\ee
An equivalent condition, analogous to Eq. \eqref{eq_Poissonsu2}, does not explicitly rely on a basis of $\su$. To display this form, we need to introduce the Maurer-Cartan forms of $\SU$. These are 1-forms $\theta^{L,R}$ with values in the Lie algebra $\su$ defined by the equations $\iota(x^{L,R})\theta^{L,R} = x$, $\forall x\in\su$. More conveniently, they can be written (with matrix groups in mind) as
\be
\theta^L \big|_H= H^{-1}\D H \qquad \theta^R\big|_H = \D H H^{-1}.
\ee
Using these formulas we can check that $\D f = \big[(\tau_k)^{L,R} f\big]\theta_k^{L,R}$, where $\theta^{L,R} = \tau_k \theta_k^{L,R} $. Then, by using the identity $x^L=(\Ad_H x)^R $, and substituting $y=\big[(\tau_k)^{R} f\big]\tau_k \in\su$, we obtain:\footnote{Note that $(\tau_k)^L f   = (\Ad_H\tau_k)^R f = (\Ad_H\tau_k)_i (\tau_i)^R f  = (\Ad_{H^{-1}}\tau_i)_k (\tau_i)^R f = [\Ad_{H^{-1}}y]_k$.}
\be
P^\#\big( K(y,\theta^R\big|_H) \big) = \frac{1}{2}\left[ (1+\Ad_{H^{-1}})y \right]_{\SU} \qquad \forall y\in\su.
\label{eq_momentmap}
\ee
From this equation and the non-degeneracy of the Maurer-Cartan forms, it is clear that the quasi-Poisson bivector $P$ has a kernel when $(1+\Ad_{H^{-1}})$ is non-invertible. In the case of $\SU$ this is when $H$ has the form $\exp (\pi \hat n.\vec\tau)$. We will return to this observation briefly.

In order to obtain a completely explicit formula for $P$, we coordinatize the group $\SU$. Coordinates on the Lie algebra are natural and allow comparison with the flat case, in particular, making the flat limit easy to evaluate, so we use the $\{a^{i}\}_{i=1}^3$ as coordinates. In the fundamental representation
\be
\bH = \exp \vec a.\vec \btau = \cos \frac{a}{2} \mathbf 1 - \I \sin \frac{a}{2} \hat n. \vec \bsigma
\ee
and convenient intermediate quantities are
\be
t_H := \Tr(\bH)=2 \cos\frac{a}{2}  \quad \text{and} \quad \vec N_H := \Tr( \bH \vec \btau ) = -\sin\frac{a}{2}\hat n\;.
\ee
By inserting $f(H)=t_H$ in Eq. \eqref{eq_PoissonSU2}, we obtain
\be
P^\#(\D t_H) & = \Tr(\bH \btau_k  )(\tau_k)_{\SU} = -\sin\frac{a}{2}n^k (\tau_k)_{\SU}.
\label{eq_traceflow}
\ee
Now, observe that the action by conjugation of the group on itself exponentiates naturally, becoming an action by conjugation at the level of the Lie algebra. Therefore, the infinitesimal version of the action $\exp {(a^G)^k}\tau_k =: H^G := GHG^{-1} = \exp \vec a. \Ad_G\vec \tau \,$ is, in our coordinates,
\be
(a^G)^j =a^i K(\tau_j, \Ad_G \tau_i )  \quad \leadsto \quad y_{\SU} a^j = a^i K(\tau_j, \ad_y \tau_i) = a^i y^k \epsilon_{jki},
\ee
and thus
\be
y_{\SU} =  a^i y^k \epsilon_{ijk}\frac{\partial}{\partial a^j} .
\label{eq_ySU2}
\ee
Substituting this into the formula for $P^\#(\D t_H)$  and using $\vec a = a \hat n$, one finds
\be
P^\#(\D t_H) = \sin \frac{a}{2} n^k a^i y^k \epsilon_{ijk}\frac{\partial}{\partial a^j} \equiv 0 \quad \Rightarrow \quad P^\#(\D a) \equiv 0.
\ee
This means that $P$ is transverse to the radial coordinate in the coordinate space.

Upon substituting $f(H)=N^i_H$ into Eq. \eqref{eq_PoissonSU2} we find,
\be
P^\#(\D N^i_H) & = \frac{1}{2} \big[\Tr(\bH \btau_k \btau_i ) + \Tr(\btau_k\bH \btau_i )  \big] (\tau_k)_{\SU} = -\frac{1}{4} \Tr(\bH) (\tau_i)_{\SU} = - \frac{1}{2} \cos\frac{a}{2} (\tau_i)_{\SU},
\ee
and from the fact that $P^\#(\D a) =0$ it is then immediate to deduce
\be
P^\#(\D a^k) = \frac{a}{2} \ctg \frac{a}{2} a^i\epsilon_{ijk}\frac{\partial}{\partial a^j}.
\ee
This gives, finally, the quasi-Poisson brackets on the group $\SU$ in terms of the logarithmic coordinates $a^k$:
\be
\big\{a^i, a^j\big\}_{qP} = \frac{a}{2} \ctg \frac{a}{2} \epsilon^{ij}_{\phantom{ij}k} a^k.
\label{eq_structureSU2}
\ee
This expression manifestly shows that the quasi-Poisson bivector is tangent to and non-degenerate on the conjugacy classes of $\Su$. This generalizes the classical result that coadjoint orbits are the symplectic leaves of the dual of the Lie algebra equipped with the canonical Kirillov-Kostant Poisson structure. This is a particular case of a more general statement about foliations of quasi-Poisson manifolds into non-degenerate leaves invariant under the group action \cite{Alekseev2002}.

At this point, one might want to introduce a rescaling of the coordinates $a^i$, to see how the flat limit appears. Consider a homogeneously curved geometry with radius of curvature $r$, then
%Geometrically, this can be done by considering geometries on a homogeneously curved geometry whose radius of curvature is $r$. Then
\be
\bH \quad\leadsto\quad ^r\bH = \exp \frac{\vec a.\vec\btau}{r^2}. 
\ee
Since this is formally obtained by sending $a \mapsto a/r^2$, Eq. \eqref{eq_structureSU2} for $r\neq1$ is
\be
\Big\{ a^i, a^j \Big\}_{qP}^r :=r^{-2} \Big\{ a^i, a^j \Big\}_{qP} = \frac{a}{2r^2}\ctg \frac{a}{2r^2} \epsilon^{ij}_{\phantom{ij}k} a^k \xrightarrow{r\rightarrow\infty} \epsilon^{ij}_{\phantom{ij}k} a^k + O(r^{-2}).
\ee
The rescaling of the quasi-Poisson brackets $\{\cdot,\cdot\}_{qP}\mapsto\{\cdot,\cdot\}_{qP}^r:=r^{-2}\{\cdot,\cdot\}_{qP}$ makes the limit clean and can be achieved by a rescaling of the Killing form appearing in the definition of $\hat y$: $K(\cdot,\cdot)\mapsto K_r(\cdot,\cdot):=r^{2}K(\cdot,\cdot)$. Interpreting the Killing form as a metric on the Lie algebra, this is equivalent to fixing its scale to that of the geometric $\mathrm S^3$ (or $\mathrm H^3$). Notice, however, that this is not a completely obvious feature, since this metric is \textit{a priori} used to measure the lengths of \textit{area} vectors, and not geometrical distances.

The quasi-Poisson structure we have just defined has various interesting features. First of all, even though the theory of group-valued moment maps that leads to Eq. \eqref{eq_PoissonSU2} generically gives quasi-Poisson brackets that \textit{violate} the Jacobi identity, in our case this does not happen. This surprise is because of the choice of group, $\SU$, and is probably not too significant; we are still forced to use genuinely quasi-Poisson spaces. In the next section it will become clear, in particular, that the ``fusion'' of four face phase-spaces cannot be performed by simple tensor product, and needs further care. Other examples of this are: the quasi-symplectic 2-form on the leaves tangent to the quasi-Poisson bivector is not simply given by the inverse of its restriction; and the formula for the quasi-symplectic volume also needs careful corrections, see \autoref{sec_qhamiltonian}. 

\section{Phase space of shapes of curved tetrahedra}\label{sec_tetrahedron_phsp}

The goal of this section is to put together the four quasi-Poisson spaces associated to the faces of a curved tetrahedron, and to subsequently reduce this quasi-Poisson space by the closure constraint $H_4H_3H_2H_1=\E$. Remarkably, the reduced space obtained by ``gluing'' multiple quasi-Poisson spaces is eventually a \textit{symplectic} space. Indeed, it is the moduli space of flat $\SU$-connections on the four-times punctured sphere
%\footnote{In this way one can actually produce the symplectic structure on the moduli space of any surface with arbitrary genus and number of boundary components.} 
equipped with the symplectic 2-form induced by the Atiyah-Bott 2-form \cite{Alekseev1998,Alekseev2000}. The ``gluing'' procedure goes under the name of \textit{fusion}, and is more complicated than in the standard case of Lie-algebra-valued moment-map theory. In the latter context it is enough to juxtapose the two Poisson manifolds each with its Poisson structure and to consider a total moment map given by the sum of the two moment maps. For examaple, in angular momentum theory, the total angular momentum is just the sum of the two angular momenta. For quasi-Poisson manifolds this is no longer possible. The total moment map should be the \textit{product} of the two moment maps, and since this operation is non-linear, one is forced to add a term to the total quasi-Poisson bivector in order to ensure the moment map condition is still satisfied in the fused space; i.e. in order to ensure that the total momentum generates the same gauge transformation on the two subspaces. In other words, a twist is needed to convert a non-linear operation (the product of two momenta) into a linear one (the sum of the two vector fields generating the gauge transformations on each copy of the group). We turn now to making this statement precise.% \textcolor{red}{I am afraid this is the first time I introduce the word ``moment map''! can it be avoided? if no, need to introduce it before (maybe this option in the best one)} \\

\paragraph*{Fusion product}
Consider two copies of the group $\SU$, i.e. the total quasi-Poisson space associated to two faces of the tetrahedron; by assumption, we require the total momentum $H_2H_1$ to be the quasi-Hamiltonian generator of gauge transformations, i.e. rigid rotations, in the total space. (Here we have in mind that we eventually want the closure constraint $H_4H_3H_2H_1=\E$ to generate rigid rotations of the full tetrahedron.) Let us now be naive and take as a quasi-Poisson bivector on the total space $P' = P_1 + P_2$, where $P_{1,2}$ are the quasi-Poisson bivectors defined on the first and second copy of $\SU$ respectively, and let us calculate the analogue of the left hand side of Eq. \eqref{eq_momentmap}:\footnote{To see that $\theta^R\big|_{H_2H_1} = \theta^R\big|_{H_2} + \Ad_{H_2}\theta^R\big|_{H_1}$, it is convenient to use $\theta^R\big|_H = \D H H^{-1}$ and apply the Leibniz rule.}
\begin{align}
{P'}^\#\big( K(y,\theta^R\big|_{H_2 H_1}) \big) & =  {P'}^\#\big( K(y,\theta^R\big|_{H_2} + \Ad_{H_2}\theta^R\big|_{H_1}) \big) \notag\\
& =P_2^\#\big( K(y,\theta^R\big|_{H_2}) \big) + P_1^\#\big( K(\Ad_{H_2^{-1}}y,\theta^R\big|_{H_1}) \big)\notag\\
& =  \frac{1}{2}\left[ (1+\Ad_{H_2^{-1}})y \right]_{\SU^{(2)}} + \frac{1}{2}\left[ (1+\Ad_{H_1^{-1}})\Ad_{H_2^{-1}}y \right]_{\SU^{(1)}},
\end{align}
where in the second step we used the linearity of $P_{1,2}$ and the $\Ad$-invariance of $K$, and in the third one we used Eq. \eqref{eq_momentmap}. This transformation has the undesirable property that it  treats the first and the second copies of the group on a different footing. 

To do better, and generalize Eq. \eqref{eq_momentmap}, this formula should involve the adjoint action associated to the product $H_2 H_1$, in both factors on the right-hand side. This is accomplished by introducing the bivector
\be
\psi_{21} := \frac{1}{2} \sum_k (\tau_k)_2 \wedge (\tau_k)_1,
\ee
where $(\tau_k)_\ell := (\tau_k)_{\SU^{(\ell)}}$ indicates the vector field generating the action by conjugation in the direction $\tau_k$ within the $\ell$-th copy of the group. Then, 
\begin{align}
{\psi_{12}}^\#\big( K(y,\theta^R\big|_{H_2 H_1}) \big) & =  {\psi_{12}}^\#\big( K(y,\theta^R\big|_{H_2}) \big) + {\psi_{12}}^\#\big( K(\Ad_{H_2^{-1}}y,\theta^R\big|_{H_1}) \big)\notag\\
& =  \frac{1}{2}\sum_k K\Big(y,\iota\left((\tau_k)_2\right)\theta^R\big|_{H_2} \Big)(\tau_k)_1 - \frac{1}{2}\sum_k K\Big(\Ad_{H_2^{-1}}y,\iota\left((\tau_k)_1\right)\theta^R\big|_{H_1} \Big)(\tau_k)_2\notag\\
%& = \frac{1}{2}\sum_k K\Big(y,(1- \Ad_{H_2})\tau_k \Big)(\tau_k)_1 - \frac{1}{2}\sum_k K\Big(\Ad_{H_2^{-1}}y,(1- \Ad_{H_1})\tau_k \Big)(\tau_k)_2\notag\\
%& = \frac{1}{2}\sum_k K\Big((1- \Ad_{H_2^{-1}})y,\tau_k \Big)(\tau_k)_1 - \frac{1}{2}\sum_k K\Big((1- \Ad_{H_1^{-1}})\Ad_{H_2^{-1}}y,\tau_k \Big)(\tau_k)_2\notag\\
& = \frac{1}{2}\Big[(1- \Ad_{H_2^{-1}})y\Big]_{\SU^{(1)}} - \frac{1}{2} \Big[(1- \Ad_{H_1^{-1}})\Ad_{H_2^{-1}}y\Big]_{\SU^{(2)}},
\end{align}
where we used $(\tau_k)_\ell = (\tau_k)_\ell^R  -(\tau_k)_\ell^L$, as well as $y^L = (\Ad_H y )^R$. This calculation shows that the correct ``fused'' quasi-Poisson bivector is
%\footnote{\textcolor{red}{ Notice that here I must put a plus instead of a minus... confusing! I inverted 1 and 2 in all the formulas. I think the inconsistency comes from the fact that my $P^\#$ is eating the FIRST entry $P$...}}
\be
P_{2\oast 1} = P_2 + P_1 + \psi_{21} \;,
\ee
because it satisfies the moment map condition
\be
P_{2\oast1}\big( K(y,\theta^R\big|_{H_2 H_1}) \big) = \frac{1}{2}\left[ (1+\Ad_{(H_2H_1)^{-1}})y \right]_{\SU^{(2)}\times\SU^{(1)}} \;,
\ee
where $y_{\SU^{(2)}\times\SU^{(1)}} = y_{\SU^{(2)}} + y_{\SU^{(1)}}$. 

Note that the fusion procedure brings the \textit{quasi}-Poisson character of these constructions to the forefront. In particular, the quasi-Poisson brackets on the product space do not satisfy the Jacobi identity, and we see that the Jacobi identity on a single copy of $\SU$ only held by a fortunate coincidence, in a sense due to the low dimensionality of the space. To be more specific, the violation of the Jacobi identity is given by:
%\footnote{A quick calculation shows that $\phi:=\frac{1}{12} \epsilon^{ijk} (\tau_i)_{\SU} \wedge  (\tau_j)_{\SU} \wedge  (\tau_k)_{\SU} \equiv 0$.}
\be
\phi_{21}:=\frac{1}{12} \epsilon^{ijk} (\tau_i)_{2\times1} \wedge  (\tau_j)_{2\times1} \wedge  (\tau_k)_{2\times1}, \qquad \text{where } (\tau_i)_{2\times1} =  (\tau_i)_2 + (\tau_i)_1 \;.
\label{eq_violationJ}
\ee
Note also, that the fusion is \textit{not} commutative, since $\psi_{12}\neq\psi_{21}$. This reflects the fact that the group product itself is non-commutative and becomes even more apparent when you iterate the process. However, It is a quick check that the fusion product is associative. Then, the total quasi-Poisson space associated to the four faces of the tetrahedron is $\SU^{\otimes4}$ equipped with the quasi-Poisson bivector
\be
P_{\oast4} := P_{4\oast3\oast2\oast1} = P_4 + P_3 + P_2 + P_1 + \psi_{21} + \psi_{31} + \psi_{41} + \psi_{32} + \psi_{42} + \psi_{43} \;.
\ee
This quasi-Poisson bivector violates the Jacobi identity by a term $\phi_{4321}$ generalizing Eq. \eqref{eq_violationJ}.

\paragraph*{Reduction} All that remains is to find the space that results upon reduction by the closure constraint 
\be
H_4H_3H_2H_1=\E \;.
\ee
By construction, the total momentum $H_4\cdots H_1$ generates rigid rotations of the tetrahedron, i.e. the diagonal conjugacy transformation: $H_\ell \mapsto G H_\ell G^{-1} \; \forall \ell$. At the end of  section \ref{sec_qpoisson} we noted that each quasi-Poisson bivector is tangent to all the conjugacy classes of the various $\SU^{(\ell)}$; this fact is unchanged by introducing the $\psi_{\ell m}$. This allows us to restrict attention to the space of shapes of tetrahedra with fixed areas, the area corresponding to the conjugacy class of the group element characterizing that face. The invariance under diagonal conjugation implies  that the coordinates on the reduced space are invariant functions under this action. And, finally, a simple counting shows that the reduced space is generically two-dimensional:
\be
4\times (\dim\SU -1) - \dim \SU_\text{closure} - \dim\SU_\text{gauge} = 2 \;,
\ee
where $(\dim\SU-1)$ is the generic dimension of a conjugacy class in $\SU$, the second term accounts for the closure constraint and the last mods out the transformations generated by this constraint. This is essentially the same counting as in the flat case. 

Then we can coordinatize the reduced phase space by any two independent conjugation-invariant functions of the $H_\ell$ (both distinct from the traces of $H_{\ell}$). The most natural choice seems to be a couple of functions of the type $\{\Tr(\bH_2 \bH_1),\Tr(\bH_4 \bH_3\}$. Had we chosen such coordinates the (quasi-)Poisson bracket between them in the reduced space would simply be the one induced by the quasi-Poisson bivector $P_{\oast 4}$: $\{ \Tr(\bH_2 \bH_1), \Tr(\bH_4 \bH_3)\}_\text{red} \equiv \{\Tr(\bH_2 \bH_1), \Tr(\bH_4 \bH_3)\}_{\oast4}$. Although this procedure for defining the reduced phase space is perfectly admissible, it does not lead to a pair of conjugate variables. To find those, it is more convenient to ask the following question: what is the (quasi-)Hamiltonian flow generated by $\Tr(\bH_2 \bH_1)$? Since the reduced space is two dimensional, the answer to this  question will immediately reveal the conjugate variable to $\Tr(\bH_2 \bH_1)$, in the form of the flow parameter.

For notational convenience let us introduce
\be
\Delta_{21}:= \Tr(\bH_2 \bH_1) = 2\cos\frac{a_2}{2}\cos\frac{a_1}{2} - 2\sin\frac{a_2}{2}\sin\frac{a_1}{2} \hat n_2 . \hat n_1.
\ee
The spherical law of cosines, Eq. \eqref{eq_sphcosinv}, immediately yields the interpretation of $\Delta_{21}$; consider the spherical triangle of $\mathrm S^3$ defined by the length of two of its edges, $a_1/2$ and $a_2/2$ respectively, and the angle subtended by them, $\arccos(\hat n_2 . \hat n_1)$. Then, the third side has a length given by $A_{21}:=\arccos(\Delta_{21}/2)$. The same construction can be repeated for $\Delta_{43}$ using side lengths given by $a_3/2$ and $a_4/2$. The closure constraint ensures that $\Delta_{21}$ and $\Delta_{43}$ have the same length and hence that the two triangles can be glued along the corresponding sides to obtain a closed polygon in $\mathrm S^3$. The angle between the two ``wings'' of the polygon can be fixed by calculating $\Delta_{14}=\Delta_{32}$, which fixes the distance between the other two vertices of the polygon.\footnote{In this way we have obtained a spherical tetrahedron in $\mathrm S^3$. This tetrahedron is in a sense ``dual'' to the one we described in the first part of the paper, its sides' lengths are equal to the areas of the that tetrahedron, which, in contrast, can be either spherical or hyperbolical. A more direct way of identifying this tetrahedron is via the identification of $\SU$ and $\mathrm S^3$: the position of the vertices are then  given by $\{\E,H_1,(H_2H_1),H_4\}\subset\SU\cong\mathrm S^3$.} 

Returning to the calculation of the flow generated by $\Delta_{21}$ we have:
\begin{align}
P_\text{red}^\#(\D\Delta_{21}) &\equiv P_{\oast4}^\# (\D\Delta_{21}) \notag\\
&= P_{2\oast1}^\# (\D\Delta_{21})+P_{4\oast3}^\# (\D\Delta_{21})+(\psi_{31}+\psi_{32})^\# (\D\Delta_{21}) + (\psi_{41} + \psi_{42})^\# (\D\Delta_{21}),
\label{eq_PdDelta}
\end{align}
where in the second equality we have grouped the terms in a convenient way. By construction $\Delta_{21}$ is the trace of the total momentum associated to the quasi-Poisson space $\SU^{(2)}\times \SU^{(1)}$. Calculating along the lines of Eqs. \eqref{eq_PoissonSU2} and \eqref{eq_traceflow}, we obtain
\be
P_{2\oast1}^\# (\D\Delta_{21}) = \Tr(\bH_2\bH_1 \btau_k)(\tau_k)_{2\times1} = -\sin\frac{ A_{21}}{2}\; \hat n_{21}^k(\tau_k)_{2\times1}\;,
\ee
where $A_{21}$ and $\hat n_{21}$ are defined by $H_2H_1 = \exp A_{21}\hat n.\vec \tau$ and again $(\tau_k)_{2\times1}=(\tau_i)_2 + (\tau_i)_1$. Meanwhile, the second term in Eq. \eqref{eq_PdDelta} vanishes immediately due to the mismatched dependencies. But, what about the last two terms? For definiteness, let us focus on the first one. Both of its sub-terms must clearly be proportional to $(\tau_k)_3$. However,  since $(\tau_k)_\ell$ is by definition the generator of conjugations of $H_\ell$ in direction $\tau_k$ and $\Tr[(\textbf{G}\bH_2\textbf{G}^{-1})\bH_1] = \Tr[\bH_2(\textbf{G}^{-1}\bH_1\textbf{G})]$ we see that $\iota\big((\tau_k)_2\big)\D\Delta_{21} = - \iota\big((\tau_k)_1\big)\D\Delta_{21}$ and this term as a whole vanishes. The final result of this computation is then
\be
P_\text{red}^\#(\D\Delta_{21}) =  -\sin\frac{ A_{21}}{2}\; \hat n_{21}^k(\tau_k)_{2\times1}\;,
\ee
which by simple derivation of the explicit expression for  $\Delta_{21}$ can also be written as
\be
P_\text{red}^\#(\D A_{21}) =   \hat n_{21}^k(\tau_k)_{2\times1}\;.
\ee

This expression has an interesting geometrical interpretation: the length of the diagonal $(21)$ of the spherical polygon generates a Hamiltonian flow that rigidly rotates  the sides 1 and 2 of the polygon around itself and leaves the sides 3 and 4 fixed. The natural parameter of this flow is the angle $\varphi_{21}$ between the wings $(21)$ and $(34)$ of the polygon hinged by the diagonal $(21)=(43)$:
\be
P_\text{red}^\#(\D A_{21}) = \frac{\partial}{\partial \varphi_{21}}\;.
\ee 
Once expressed this way, it is also clear that this is a gauge invariant statement that makes perfect sense in the reduced space where there is no difference between the diagonal $(21)$ and $(43)$. This Hamiltonian flow is the simplest instantiation of a \textit{bending flow} \cite{Kapovich1996,Kapovich2000,Treloar2000}. Locally the same result holds in the phase space of flat tetrahedra \cite{Baez1999, Bianchi2011a, Bianchi2012b}. Notice however, that the global structure of the Poisson space is very different: in particular, the interval in which $A_{21}$ lives at fixed (large enough) $a_1$ and $a_2$ is generically modified by the compact nature of the three-sphere. This has important consequences for the quantization of this space.

The Poisson bivector
\be
P_\text{red} = \frac{\partial}{\partial A_{21}} \wedge \frac{\partial}{\partial \varphi_{21}}  \qquad \Leftrightarrow  \qquad \big\{A_{21},\varphi_{21}\big\}_\text{red} = 1
\ee
is a completely standard Poisson structure with no trace of quasi structure. This is not a coincidence, since it is a general feature of the reduced phase spaces of this kind that they are Poisson (in fact, symplectic) spaces. Taking our concrete case as an example, this can be understood from the expression of $\phi_{4321}$, the term encoding the violation of the Jacobi identity in the total quasi-Poisson space before reduction. This is a tri-vector composed of terms generating the diagonal conjugacy transformation in the four copies of $\SU$. However, the reduced space is obtained precisely by requiring invariance under such transformations.

The symplectic coordinates $(A_{21},\varphi_{21})$ relate to the complex Fenchel-Nielsen (FN) coordinates $(x,y)$ of flat connections on a four-punctured sphere \cite{kabaya}, satisfying $\{\ln x,\ln y\}=1$. The complex FN length variable $x$ is the eigenvalue of the holonomy along a loop encircling two punctures, i.e. the eigenvalue of $H_2H_1$. Hence, $x^2=\exp(-\I A_{21})$. On the other hand, $\varphi_{21}$ is the conjugate variable, and is therefore related to the (logarithmic) complex FN twist variable $\ln y$ up to a certain function of $A_{21}$.

As a last remark, we point out an explicit expression for $\varphi_{21}$:
\be
\varphi_{21} = \arccos\left( 
\frac{\hat n_1 \times \hat n_{21}}{|\hat n_1 \times \hat n_{21}|} \cdot \frac{\hat n_4 \times \hat n_{21}}{|\hat n_4 \times \hat n_{21}|}
\right).
\ee

%-----------------------------------------------------------------------------------------
\section{Quasi-Hamiltonian approach}\label{sec_qhamiltonian}

In this section we give a very brief account of the quasi-Hamiltonian approach to the phase space of shapes. Our main goal is to calculate the quasi-symplectic volume (area) of the leaves. In this formulation one is forced to work directly at the level of the conjugacy classes, i.e. on the leaves. This is because the quasi-symplectic two form is in a sense the inverse of the quasi-Poisson bivector, and as such can't have any degenerate direction. This statement would apply precisely in the standard symplectic case corresponding to the flat limit in which the group elements are substituted by Lie-algebra elements. However, as is often the case, in the quasi setting there are important twists to the original definitions. 

We first recall the standard symplectic structure on the coadjoint orbits in $\fg^*$. The coadjoint orbit $\mathcal O_\alpha$ of an element $\alpha\in\su^*$ is defined as
\be
\mathcal O_\alpha =\left\{ \beta \in \su^* \, | \, \exists G \in \SU \;\text{with }\; \beta = \Ad_{G^{-1}}^*\alpha \right\}.
\ee
This set carries a canonical, closed, non-degenerate 2-form $\omega_{ \alpha }$ defined by
\be
\omega_\alpha(y_{\mathcal O_\alpha},z_{\mathcal O_\alpha}) = \langle \alpha, [y,z] \rangle \qquad \forall y,z\in\su,
\ee
where on the left-hand side $y_{\mathcal O_\alpha},z_{\mathcal O_\alpha} \in \mathfrak X(\su^*)$ are the vector fields generating the coadjoint action  in the directions $y,z\in\su$, respectively. To show that this form is closed, introduce the symbol $\text{cyclic}_{x,y,z}\{ \cdot\}$ for a summation on cyclic permutations of the elements $x,y,z$, and calculate
\begin{align}
\D \omega_\alpha(x_{\mathcal O_\alpha},y_{\mathcal O_\alpha},z_{\mathcal O_\alpha}) &= 
\text{cyclic}_{x,y,z}\left\{ x_{\mathcal O_\alpha} \omega_\alpha(y_{\mathcal O_\alpha},z_{\mathcal O_\alpha})  \right\} -
\text{cyclic}_{x,y,z}\left\{ \omega_\alpha ([x_{\mathcal O_\alpha},y_{\mathcal O_\alpha}], z_{\mathcal O_\alpha}) \right\} \notag\\
%
%&= \text{cyclic}_{x,y,z}\big\{\langle -\ad^*_x\alpha , [y,z] \rangle \big\} - \text{cyclic}_{x,y,z}\big\{\langle \alpha , [[x,y],z] \rangle \big\}\notag\\
%
&=2\text{cyclic}_{x,y,z}\big\{ \langle \alpha, [x,[y,z]] \rangle \big\} \equiv 0,
\end{align}
with the last expression vanishing due to the Jacobi identity on $\su$. In this setting, the moment map condition corresponding to Eqs. \eqref{eq_Poissongenerator} and \eqref{eq_Poissongenerator2} is:
\be
\iota(y_{\mathcal O_\alpha}) \omega_\alpha =  \langle \D\alpha, y\rangle.
\ee

We want to generalize this equation to the case where the variables are in the group instead of in the (dual of the) Lie algebra. To this end, define 
\be
\mathcal O_H =\left\{ H' \in \SU\, |\, \exists G \in \SU \;\text{s.t.}\; H'= G H G^{-1} \right\} \;,
\ee
and for any $y\in\su$ define $y_{\mathcal O_H}$ as the vector field that generates the conjugation action  on $\mathcal O_H$ in the direction $y$ (these are simply the restriction to $\mathcal O_H\subset \SU$ of the $y_{\SU}=y^R-y^L$ defined above). Then, the generalization of the quasi-Hamiltonian moment map condition reads
\be
\iota(y_{\mathcal O_H})\omega_H = \frac{1}{2} K\left(\theta^L + \theta^R\big|_H  , y \right)\;,
\ee
where we have denoted the quasi-symplectic two form on the conjugacy class $\mathcal O_H \subset \SU$ by $\omega_H$ and in this section the $\big|_H$ is used as shorthand for the pullback  of the Maurer-Cartan forms to  $\mathcal O_H$. This formula can be justified very similarly to its counterpart in the quasi-Poisson construction: the 1-form $\D \alpha$ is substituted by the a particular combination of left and right Marurer-Cartan forms, i.e. $\frac{1}{2}(\theta^R + \theta^L)$, that is compatible with the antisymmetry of $\omega_H$:
\be
\iota(y_{\mathcal O_H}\otimes y_{\mathcal O_H})\omega_H = \frac{1}{2} K\left(\iota(y_{\mathcal O_H})(\theta^L + \theta^R)\big|_H  , y \right) \equiv 0.
\ee
One can show, see \cite{Alekseev1998}, that the above moment map condition implies that the form $\omega_H$ is not closed:
\be
\D \omega_H = - \frac{1}{12} K\left(\theta^L,[\theta^L,\theta^L]\right)\big|_H \;,
\ee
and therefore it is \textit{not} symplectic. Moreover, it has a kernel on the equatorial region of $\SU$:
\be
\mathrm{ker} \omega_H = \left\{ y_{\mathcal  O_H}\, | \, y\in\su \text{ and } y\in \mathrm{ker}(\Ad_H +1) \right\} .
\ee
This is analogous to the presence of a kernel for the quasi-Poisson bivector $P$. Notice that the two are \textit{not} inverses of one another, as in the standard symplectic case,
%\footnote{At least after the appropriate restriction on the domain of $P$ are taken into account} 
and the relation between them is more involved. We refer to \cite{Alekseev2002} for details, but for completeness we provide the inversion formula
\be
P_{\mathcal O_H}^\# \circ \omega_H^\flat = \mathrm{Id}_{\mathrm T \mathcal O_H} - \frac{1}{4} (\tau_k)_{\mathcal O_H} \otimes (\theta^L_k - \theta^R_k) \big|_H
\ee
where $\theta^L_k = K(\theta^L,\tau_k)$. 

One can also perform the fusion of two quasi-symplectic spaces, and again this procedure needs a twist, which reads:
\be
\omega_{2\oast1} = \omega_2 + \omega_1 + \frac{1}{2} \sum_k \left(\theta^L_k\right)_2\wedge\left(\theta^R_k\right)_1\;.
\ee
Quasi-Hamiltonian reduction is possible as well, and leads to a standard symplectic space, much as reduction in the quasi-Poisson setting did. 

Finally, we want to mention that it is possible to associate a volume to the quasi-symplectic spaces $\mathcal O_H$, which can eventually be used to calculate the volume of the reduced space, leading to the expected result; that is, to Witten's formula for the symplectic volume of the moduli space of flat connections on a Riemann surface, see \cite{Alekseev2002b} and references therein. We do not go into this topic in any detail, but simply calculate the volume of  a single leaf. For this we need an explicit expression for $\omega_H$. To calculate this, we turn to the moment map condition, and contract it with another vector field $z_{\mathcal O_H}$:
\begin{align}
\iota(y_{\mathcal O_H}\wedge z_{\mathcal O_H})\omega_H & = \epsilon^i_{\phantom{i}jk} y^j a^k \epsilon^l_{\phantom{l}mn} z^m a^n \iota\left( \frac{\partial}{\partial a^i}\wedge\frac{\partial}{\partial a^l} \right)\omega_H\notag\\
& = \frac{\sin a}{a} \epsilon_{pqr}y^p z^q a^r,
\end{align}
where we have again parametrized $\SU$ by $H=\exp \vec a.\vec\tau$ with $\vec a = a \hat n$ and in the first line we used the explicit expression of $y_{\SU}$ obtained in Eq. \eqref{eq_ySU2}. In the second line we used the following explicit formula for $\frac{1}{2}(\theta^L+\theta^R)$:
\be
\frac{1}{2}\left(\btheta^L + \btheta^R\right)\big|_H = \left[ \frac{\sin a}{a} \delta^i_j + \frac{a - \sin a}{a} n^i n_j\right] \btau_i \D a^j\;.
\ee
Because $y$ and $z\in\su$ are arbitrary, it follows that:
\be
\omega_H = \frac{\sin a}{a} a^k\epsilon_{kij} \;\D a^i \wedge \D a^j \;.
\ee
This formula should be understood as restricted to the conjugacy class of $H$, that is to the sphere of radius $a$ within the coordinate space $\{\vec a\}$. This is consistent, since the 2-form of the previous formula is tangent to these spheres, in the sense that it vanishes when contracted in the radial direction: $\iota(\partial/\partial a)\omega_H \equiv 0$. Anyway, to make this fact completely explicit, it suffices to recognize that on the 2-sphere of radius $a$, $\mathcal O_a$, the quasi-symplectic 2-form is just
\be
\omega_a = \sin a \;\D^2\Omega,
\ee
where $\D^2\Omega$ is the homogeneous measure on the unit 2-sphere. This can be compared with the symplectic form on the $\su^*$ coadjoint orbits $\omega^{\su^*}_a = a\, \D^2\Omega$. Notice that from this formula it is evident that $\omega_a$ happens to be closed. This is the quasi-symplectic version of the fact that $P_{\SU}$ happens to satisfy the Jacobi identity. Here it is even more clear that this happens for purely dimensional reason: the leaves $\mathcal O_a$ are 2-dimensional and therefore $\omega_a$ is already a top-dimensional form. This would not happen for other groups, nor in the fusion space of two or more leaves. Also, note that $\omega_a$ vanishes at $a=\pi$, i.e. exactly were the operator $(1+\Ad_H)$ has a kernel.

In the theory of quasi-symplectic spaces, the generalization of the Liouville form $\mathcal L$ has an extra term in order to assure that $\mathcal L\neq 0$ everywhere. This generalization can be used to calculate the symplectic volume of the moduli space of flat connections (Witten's formulas) \cite{Alekseev2002b}. The generalized expression of $\mathcal L$ in our context is
\be
\mathcal L_H = \frac{\omega_H}{\sqrt{ \det\left( \frac{1 + \Ad_H}{2}\right) }}.
\ee
To calculate the determinant we observe that the adjoint action of $H\in\SU$ on its Lie algebra is essentially an action by rotation around the axis $\hat n$ by  an angle $a$. Moreover, the determinant is invariant under conjugations of its argument, and therefore the axis $\hat n$ can be fixed to the $\hat z$-axis. In this way
\be
\det\left( \frac{1 + \Ad_H}{2}\right) = \det\left( \frac{1 + R_z(a)}{2}\right) \qquad \text{where}\qquad
 R_z(a)=\left(
\begin{array}{ccc}
\cos a & -\sin a & 0 \\
\sin a & \cos a & 0 \\
 0 & 0 & 1
 \end{array}
 \right).
\ee
This immediately gives
\be
\mathcal L_H  = 2\sin\frac{a}{2}\; \D^2\Omega\;.
\ee
Notice that in the study of the coadjoint orbits the Liouville measure is identical to $\omega_H^{\su^*}$ with no extra factor needed. Therefore, this formula should be compared to $\mathcal L^{\su^*}_\alpha = a \,\D^2\Omega$. Wonderfully, $\mathcal L_H$ is totally regular at $a=\pi$, i.e. where $\omega_a$ was found to be vanishing. The volume form $\mathcal L_H$ vanishes only at $a=2\pi$, precisely on the only non-trivial central element of $\SU$, where the orbit $\mathcal O_{H=-\E}$ reduces to a point. The expression for $\mathcal L_H$ most clearly displays the compact nature of the area-vector spaces. This compactness has important consequences for the quantization of these systems. See section \ref{sec_outlook} for a brief discussion of this. 

To conclude, we provide expressions for $\omega_H$ and $\mathcal L_H$ explicitly displaying the radius of curvature $r$:
\be
\omega_H = r^2 \sin \frac{a}{r^2}\;\D^2\Omega \qquad \text{and} \qquad \mathcal L_H = 2 r^2 \sin \frac{a}{2r^2} \;\D^2\Omega.
\ee
It is clear that in the limit $r\rightarrow \infty$ one recovers  $\omega^{\su^*}_a$ and $\mathcal L^{\su^*}_\alpha$, with no residual  dependence on $r$.

%-----------------------------------------------------------------------------------------
\section{Summary}\label{sec_summary}

Minkowski's theorem establishes a one-to-one correspondence between closed non-planar polygons in $\rmE^3$ and convex polyhedra, via the interpretation of the vectors defining the sides of the polygon as area vectors for the polygon. Extending this theorem to curved polyhedra is non-trivial. We have proven the first generalization, to the best of our knowlege, of Minkowski's theorem for curved tetrahedra. Our techniques, make it possible, at least in principle, to extend the result to more general polyhedra. Our theorem establishes a correspondence between non-planar, geodesic quadrilaterals  in $\rmS^3$, encoded in four $\SO$ group elements $\{O_\ell\}$ whose product is the group identity, and flatly embedded tetrahedra in either $\rmS^3$ or $\rmH^3$. This correspondence depends on the choice of a (non-canonical) isomorphism between the fundamental groups of the four-punctured two-sphere and the tetrahedron's one-skeleton. Finally, we used our theorem to reinterpret the Kapovich-Millson-Treloar symplectic structure of closed polygons on a homogeneous space (with fixed side lengths and up to global isometries), as the symplectic structure on the space of shapes of curved tetrahedra (with fixed face areas and up to global isometries). 

In the context of the phase space construction, it was important to lift the $\{O_\ell\}$ to elements $\{H_{\ell}\}$ of $\Su$. Because $\Su$ is a double cover of $\SO$, a given tetrahedron is not in one-to-one correspondence with four $\Su$ group elements multiplying to the identity. Nonetheless, bijectivity can be recovered if one decorates the sides of the tetrahedron with plus and minus signs. These can be thought of as relative orientations of the reference frames at the various vertices of the tetrahedron, and are corrections that a spinor would be sensitive to. In this sense, these signs are the extra structure one would expect to need to have a full description of a discrete spinorial geometry. A difficulty that arises in this context is that there is no way, in general, to extend a spin structure from the one-skeleton of the tetrahedron to the full ambient space consistently. Nevertheless, it is intriguing that the lift from a Levi-Civita to a spin connection is required to effectively treat the symplectic nature of a tetrahedron's shape.

The geometrical construction investigated here also gave rise to a couple of other unexpected features. First of all, one does not need two distinguished frameworks to deal with spherical and hyperbolic geometries, as in two dimensions \cite{Bonzom2014a,Dupuis2014,Charles2015}. On the contrary, four $\SO$ group elements that close encompass both scenarios. This can be interpreted as follows: The four $\SO$ group elements are precisely four Levi-Civita parallel transport holonomies that an observer might measure by following a (topological) tetrahedron's one-skeleton in some general Riemaniann space. If so, these four holonomies are all the observer knows about that region of space; what is the best approximation she can give of the geometry of that region with the information at her disposal? We claim this is the tetrahedral geometry our theorem allows her to reconstruct. With this picture in mind, the choice of isomorphism between the fundamental groups of a four-punctured sphere and the tetrahedron's one skeleton is not an extra ingredient, but simply a datum arising in her experimental setting. There are, however, situations in which the previous picture fails to be viable. This is where the second unexpected feature of the theorem comes in: at times one encounters geometries which are hyperbolic and contain nevertheless triangles of area larger than $\pi$. This forced the introduction of a new type of hyperbolic triangle (and more generally, simplices) extending across the two sheets of a two-sheeted hyperoboloid. These triangles have an infinite metrical area, but finite holonomy area.

The non-commutative nature of the generalized area vectors $\{O_\ell\}$, as well as the compactness of their domain of definition, led us to consider the quasi-Poisson manifolds of Anton Alekseev and collaborators. Quasi-Poisson manifolds generalize Poisson manifolds by allowing for group-valued moment maps and a (controlled) failure of the Jacobi identity. The group-valued moment map was particularly valuable in the present work where it allowed us to generalize the fact that the closure constraint generates rigid rotations of the polyhedron to the curved context. Analogues of these facts were known in the context of the construction of the symplectic form on the moduli space of flat connections on a Riemann surface. Indeed, a finite dimensional derivation of this symplectic structure was one of the main motivations behind the work of Alekseev and collaborators. We have provided an interpretation of these results that allows for new connections between the study of flat connections on Riemann surfaces, deformed spin-networks for quantum gravity, and discrete three-dimensional curved geometries.

%-----------------------------------------------------------------------------------------
\section{Outlook}\label{sec_outlook}

We have already applied the results of this paper to the construction of a spinfoam model for four-dimensional quantum gravity with cosmological constant \cite{HHKR}. The tetrahedra described here constitute the boundary states of the model, thus this model provides a physical motivation for studying their symplectic structure. The theorem presented in this paper also serves as the foundation for the reconstruction of the (semiclassical) geometry of a curved four-simplex considered in that work. There the tetrahedron's closure relation stems from a flatness condition for the holonomies living in $\rmS^3\setminus \Gamma_5$, where $\Gamma_5$ is the four-simplex one-skeleton. Interestingly, the four-simplex does not have an analogue of the closure constraint, but instead a new set of spacetime holonomies encodes how to glue the five tetrahedra into a simplex.  For more details we refer to the cited work.

As the application to spinfoams shows, it is interesting to seek extensions of our work not only towards more general polyhedra, but in particular to a generic triangulated manifold. The most natural setting for this is that of (discrete) twisted geometries \cite{Freidel2010twist,Freidel2010twist2}. In a twisted geometry, the face shared by two polyhedra has a well defined area, but not a well defined shape, since this depends on which side it is viewed from. These geometries are a classical interpretation of spin-networks states. The spin-network is a graph colored by an $\Su$ irrep (spin) on each link and by $\Su$ intertwiners on the vertices, much like in lattice gauge theory. Each vertex is interpreted as a polyhedron, with as many faces as coincident links. These polyhedra are described by the intertwiner quantum number, and their faces carry areas encoded by the spins associated to the links. In the dual representation, one can associate to each spin-network a wave function which depends on one $\Su$ group element per link and which is invariant under $\Su$ transformations at each vertex of the graph (see \cite{Rovelli2007}). 

In this representation the group elements are interpreted as the parallel transports between the reference frames of two adjacent polyhedra, while the $\su$ generators of the $\Su$ transformation at the end point of each link are interpreted as the area vectors of the polyhedron sitting at the given vertex. Holonomies and area vectors, can be packaged into a natural symplectic structure at every (half-)link, that of $\mathrm{T}^*\Su$.  Interestingly this symplectic structure is  induced by that of general relativity when canonically quantized in Ashtekar's variables \cite{Thiemann2004}.

What our work suggests is to generalize this construction to curved, classical and quantum, twisted geoemetries. For this, one has to consider many curved tetrahedra, connected one to another by pathes colored by $\Su$ holonomies. The holonomies considered in this paper rather play the r\^ole of the area vectors, or fluxes in the spin-network parlance. This way, one ends up studying the double $\Su\times\Su$ as the pertinent generalization of $\mathrm{T}^*\Su$. Again, this space carries a natural quasi-Poisson structure, but no symplectic structure (a consequence of the triviality of its second deRahm cohomology), which was introduced by Alekseev to provide a finite dimensional construction of the moduli space of flat connections on a Riemann surface. In our context, the relevant Riemann surface is that of the thickened spin-network graph. Quantization of these deformed twisted geometries requires representations of the quantum (in the algebraic sense) objects associated to quasi-bialgebras, which are the infinitesimal analogues of the double group $\Su\times\Su$. These quantum objects are built out of a quasi-Hopf  analogue of the appropriate qunatum Lie algebra. Their representations are known to be related to those of a quantum group evaluated at a root of unity \cite{Chari1995}. The study of quantum twisted geometries is work in progress.

The phase space's compactness and the associated emergence of quantum group representations at the root of unity, which allow only a finite total number of states, have two compelling consequences: the first of these is that geometrical observables, such as the volume of a curved tetrahedron, have discrete and bounded spectra. The second is that spinfoam models built out of quantum group representations must be finite, in the sense that they have no bubble divergences \cite{han20114,fairbairn2012}. Unfortunately, we do not yet know how to make this precise in the context of our spinfoam model \cite{HHKR}, which is defined somewhat formally in terms of a complex Chern-Simons theory. The finiteness of a spinfoam model does not mean that bubbles are not potentially associated with large amplitudes (scaling with powers of the inverse cosmological constant) and may require a renormalization procedure, see \cite{Riello2013}.  From a physical perspective the Planck scale regularizes ultraviolet divergences, while the cosmological scale regularizes infrared divergences. 

Finally, an interesting application of deformed spin-networks would be to introduce a more robust coarse graining procedure for spin-networks. Indeed, the deformed networks carry local curvature at their vertices, and therefore could be used to deal with the failure of gauge invariance that standard coarse-grained spin-networks suffer from \cite{Livine2014coarse,Dittrich2015}.

\acknowledgments

The authors wish to thank Wojciech Kami\'nski for enthusiastic discussions on the geometrical features of the reconstruction during the early days of this project. \\

HMH acknowledges support from the National Science Foundation (NSF) International Research Fellowship Program (IRFP) under Grant No. OISE-1159218. MH acknowledges funding received from the People Programme (Marie Curie Actions) of the European Union’s Seventh Framework Programme (FP7/2007-2013) under REA grant agreement No.   298786.   MH also acknowledges the funding received from Alexander von Humboldt Foundation.  HMH and AR were supported by Perimeter Institute for Theoretical Physics.
Research at Perimeter Institute is supported by the Government of Canada through Industry Canada and by the Province of Ontario through the Ministry of Research and Innovation.

%
%
%%------------------------------------------------------------------------------------------
%%
%%------------------------------------------------------------------------------------------
%
%
%\newpage
%
%
%\section*{Introduction}
%
%\begin{itemize}
%\item Minkowski theorem in flat space
%\item Simple Path on tets (cite Alexandrov's book)
%\item punctured surfaces (and CS)
%\item spherical case
%\item hyperbolic case
%\item spin lift
%\item phase space of tet shapes
%\end{itemize}
%
%\section*{Theorems}
%
%\begin{itemize}
%\item follow intro
%\item Gram
%\item $n$-point functions
%\item cohomology \& spin-lift
%\end{itemize}
%
%\section*{Applications}
%
%\begin{itemize}
%\item quasi-Poisson Lie groups
%\item reduction etc.
%\item new observables
%\end{itemize}

\bibliographystyle{bibstyle_aldo}
\bibliography{0CurvedTet.bib}

\end{document}